\documentclass[reqno]{amsart}
\usepackage[utf8]{inputenc}
\usepackage[T1]{fontenc} % Use 8-bit encoding that has 256 glyphs
\usepackage{mathtools}
\usepackage{bbold}
\usepackage{hyperref}     
\theoremstyle{definition}
 
\theoremstyle{plain}
\newtheorem{theorem}{Theorem}

\newtheorem{lemma}{Lemma}
\newtheorem{cor}{Corollary}

\newcommand{\supp}{\operatorname{supp}}

\newcommand{\dl}{\mathrm{d}x\,}
\newcommand{\spr}[1]{\mathrm{tr}\left(#1\right)}
\newcommand{\dd}[1]{\mathrm{d}{#1}\,}
\title[The maximal excess charge for a family of RDMFT]{The maximal excess charge for a family of density-matrix-functional theories including Hartree-Fock and M\"uller theories}
\author[C.\ Kehle]{Christoph Kehle}
\date{\today}
\address{Mathe\-matisches
	Institut der Ludwig-Maximilians-Universit\"at M\"unchen \\ Theresienstra\ss e 39, 80333
	M\"unchen\\ Germany}
\curraddr{DPMMS, CMS, University of Cambridge \\ Wilberforce Road, Cambridge CB3 0WB, UK (as of Oct. 2016)}
\email{christoph.kehle@gmail.com}
\begin{document}
\begin{abstract}
We will give a proof that the maximal excess charge for an atom described by a family of density-matrix-functionals, which includes Hartree-Fock and M\"uller theories, is bounded by an universal constant. We will use the new technique introduced by Frank et al \cite{frank2016ionization}.
\end{abstract}
\maketitle
\tableofcontents
\section{Introduction}
A proof of the experimental fact that atoms can at most bind one extra electron is a major challenge in mathematical physics. Even a proof of the weaker bound $Z+C$ for the maximal electron number is still an open question in full Schr\" odinger theory and is known as the \textit{ionization conjecture}. 

Since full Schr\"odinger theory is analytically and numerically very complicated, approximate but simpler theories are often used to study atoms. One of the most accurate but still fairly simple approximate theories is the Thomas-Fermi-Dirac-Weizs\"acker theory, for which the ionization conjecture was proved very recently \cite{frank2016ionization}. Extending the method in \cite{frank2016ionization} and using Solovej's bootstrap argument in \cite{solovej2003ionization}, Frank, Nam and Van Den Bosch were able to provide a proof of the ionization conjecture for the more involved M\"uller theory \cite{frank2016maximal}, which relies - just like the Hartree-Fock functional - on one-particle density matrices rather than merely on electron densities. 

We shall see that this method can be used as well to prove the ionization conjecture for a family of density-matrix theories including M\"uller and Hartree-Fock theories. For any parameter $p\in[1/2,1]$ we  consider the \textit{Power functional}
\begin{align*}
&\mathcal{E}^p_Z (\gamma)= \spr{-\Delta \gamma- \frac{Z}{|x|} \gamma}  + D[\rho_\gamma] - X(\gamma^p), \\&
 D[\rho_\gamma] = \frac 12 \iint \frac{\rho_\gamma(x) \rho_\gamma(y)}{|x-y|}\dl\dd{y}, \,\;\; X(\gamma^p)= \frac 12 \iint \frac{|\gamma^p(x,y)|^2}{|x-y|}\dl\dd{y},
\end{align*}
which was introduced by Sharma et al \cite{sharma2008reduced}.  Note that $ p \in [1/2,1]$ interpolates between the M\"uller functional $\mathcal{E}^{\mathrm{M}}_Z (\gamma) = \mathcal{E}^{1/2}_Z (\gamma)$ and the Hartree-Fock functional $\mathcal{E}^{\mathrm{HF}}_Z (\gamma) = \mathcal{E}^{1}_Z (\gamma)$. 

At this point we want to motivate the choice of the exchange term $X(\gamma^p)$. By Lieb's variational principle, the ground state energy of the Hartree-Fock functional gives an upper bound for the Schr\"odinger ground state energy $E_S(N,Z)$. In \cite{frank2007muller} it is conjectured (indeed proven for $N=2$) that the ground state energy of the M\"uller functional is a lower bound of $E_S(N,Z)$. Numerical results also support this conjecture. Thus, it is no surprise that theories interpolating between these functionals give good numerical results and get more and more popular among theoretical chemists (e.g.\ \cite{kamil2015reduced,lathiotakis2009density,putaja2016validity}). Recall that the ground state energies of both, the M\"uller functional and the Hartree-Fock functional of a neutral atoms agree with the quantum ground state energy $E_S(Z,Z)$ to order $o(Z^{5/3})$\cite{bach1992error,siedentop2009asymptotic}. Thus, the same correct asymptotic behavior holds true for the Power functional.

For any parameter $p\in [1/2,1]$ we will consider the minimization problem
\begin{align}
E^p(N,Z) :=\inf\{\mathcal{E}^p_Z(\gamma) : \gamma\in \mathcal{I}, \spr{\gamma}=N \}\label{eq:minimizationprob}.
\end{align}

Here, $\mathcal{I}$ are fermionic one-particle density matrices, i.e., \begin{align*}\mathcal{I} :=\{\gamma\in \mathfrak{S}_1(L^2(\mathbb{R}^3)): 0\leq \gamma\leq 1, \Delta\gamma\in\mathfrak{S}_1(L^2(\mathbb{R}^3))\} \end{align*}
where $\mathfrak{S}_1(L^2(\mathbb{R}^3))$ denotes the trace class operators acting on $L^2(\mathbb{R}^3)$. The density is given by $\rho_\gamma(x) = \gamma(x,x)$, which can be made rigorous using the spectral decomposition of $\gamma$. 

If not stated differently, from now on, $p$ will be any number in $[1/2,1]$. All constants will be independent of $p$. Our main theorem will be 
\begin{theorem}[Ionization bound]\label{thm:ionization}
	There is a constant $C>0$ such that for all $Z>0$, if the minimization problem $E^p(N,Z)$ in \eqref{eq:minimizationprob} has a minimizer, then $N\leq Z+C$.
\end{theorem}
The proof of this theorem works in the same manner as in \cite{frank2016maximal} for the M\"uller functional. Due to the fractional operator power $\gamma^p$ it is slightly more involved. The additional technical problems arising for $1/2 < p < 1$ are proven in Section~\ref{sec:powerfunct}. Apart from this, the main strategy is to compare with Thomas-Fermi theory as in the proof for the Hartree-Fock theory \cite{solovej1991proof,solovej2003ionization}. This is captured in 
\begin{theorem}[Screened potential estimate]
	\label{thm:mainthm2} Let $N\geq Z \geq 1$ and let $\gamma_0$ be a minimizer for $E^p(N,Z)$. Let $\rho^{\mathrm{TF}}$ be the Thomas-Fermi minimizer with $\int\rho^{\mathrm{TF}} = Z$. For every $r>0$, define the screened potentials by\begin{align*}
	\Phi_r(x) = \frac{Z}{|x|} - \int_{|y|\leq r} \frac{\rho_0(y)}{|x-y|}\dd{y}, \;\; \Phi_r^{\mathrm{TF}} = \frac{Z}{|x|}-\int_{|y|\leq r}\frac{\rho^{\mathrm{TF}}(y)}{|x-y|} \dd{y}.
	\end{align*}
	Then there are universal constants $C>0, \epsilon>0$ such that \begin{align*}
	\left| \Phi_{|x|}(x) - \Phi_{|x|}^{\mathrm{TF}}(x) \right| \leq C ( |x|^{-4 +\epsilon} + 1)
	\end{align*}
for all $|x|>0$.
\end{theorem}
The significance of the power $|x|^{-4+\epsilon}$ is that $\Phi_{|x|}^{\mathrm{TF}} \sim |x|^{-4}$ for small $|x|$. 

Similar to \cite{solovej2003ionization,frank2016ionization,frank2016maximal}, we have the following asymptotic estimate for the radii of ``infinite atoms''.
\begin{theorem}[Radius estimate]\label{thm:mainthm3}
Let $\gamma_0$ be a minimizer of $E^p(N,Z)$ for some $N \geq Z \geq 1$. For $\kappa>0$, we define the radius $R(N,Z,\kappa)$ as the largest number such that\begin{align*}
\int_{|x| \geq R(N,Z,\kappa)} \rho_{\gamma_0}(x) \dl = \kappa. \end{align*}
Then there are universal constants $C>0,\epsilon>0$ such that \begin{align*}
\limsup_{N\geq Z \to \infty} 
\left| R(N,Z,\kappa) - B^{\mathrm{TF}} \kappa^{-1/3} \right| \leq C \kappa^{-\frac 13 - \epsilon}
\end{align*}
for all $\kappa\geq C$, where $B^{\mathrm{TF}} = 5 c^{\mathrm{TF}} \left( \frac{4}{3\pi^2}
\right)^{1/3}$.
\end{theorem}
Theorem~\ref{thm:ionization} and Theorem~\ref{thm:mainthm3}  will be direct consequences of  Theorem~\ref{thm:mainthm2}. To prove Theorem~\ref{thm:mainthm2} we use Solovej's bootstrap argument. As in \cite{frank2016maximal}, the ``multiplying by $|x|$'' strategy is not working. This strategy will be replaced - as in \cite{frank2016ionization} and \cite{frank2016maximal} - by a method in which $\mathbb{R}^3$ will be split into half-planes followed by an averaging process, cf.\ Section~\ref{sec:l1est}.

Having non-existence of a minimizer for $N\geq Z+C$, the natural question of existence for a minimizer for $N\leq Z$ arises. So far, this is open. In \cite{kehlemasterthesis2016} it was shown that for any $N>0$ and $Z> 1/2$, the renormalized Power functional \begin{align*}\hat{\mathcal{E}}^p_Z (\gamma):= \mathcal{E}^p_Z(\gamma) - E^p(\spr{\gamma},0)\end{align*} possesses a minimizer varying over $\spr{\gamma} \leq N$. The same method as in \cite{frank2007muller} was used. However, a proof for the existence of a minimizer for $E^p(N,Z)$ for $N\leq Z$ was not given. 

	\textit{Convention.} Throughout the paper we will assume that $E^p(N,Z)$ has a minimizer $\gamma_0$ for some $N\geq Z$. The corresponding density will be denoted by $\rho_0 = \rho_{\gamma_0}$. Note that in contrast to a minimizer of the M\"uller functional ($p=1/2$), $\rho_0$ need not to be spherically symmetric since the convexity of the functional is lost for $p>1/2$. 
\\	~\\
\textbf{Acknowledgement.} I gratefully thank R.~L.~Frank and P.\ T.\ Nam for the early communication of the results in \cite{frank2016maximal} as well as H.~Siedentop for the very helpful advice and many fruitful discussions.
\section{The Power Functional}\label{sec:powerfunct}
We will start by proving properties of the Power functional.
\subsection{General facts}
First, we like to note that the ground state energy $E^p(N,Z)$ is non-decreasing in $p$. This can be seen by writing the exchange correlation term as \begin{align}\label{eq:fefferman}
X(\gamma^p) = \frac 12 \int_\Lambda \mathrm{tr}( \gamma^p B_\lambda \gamma^p B_\lambda^\ast ) \dd{\lambda},
\end{align}
where $\Lambda$ is a parameter space. This formula can be derived using Fefferman-de la Llave formula \cite[page 4]{frank2007muller}. Comparing $E^p(N,Z)$ with the Hartree-Fock energy \begin{align*}0\geq E^{\mathrm{HF}}(N,Z) = E^1(N,Z)\geq E^p(N,Z)\end{align*} already shows that $E^p(N,Z)\leq 0$ for any $p\in[1/2,1]$, $N>0, Z\geq 0$. Indeed, using subadditivity of the ground state energy for free electrons and a scaled hydrogen minimizer, it can be shown that $E^p(N,0)<0$ for any $1/2\leq p<1$ \cite{kehlemasterthesis2016}, whereas $E^{\mathrm{HF}}(N,0)=0$. This means that free electrons have negative (binding) energy for $1/2\leq p <1$.

Now, we want to prove that $E^p(N,Z)$ is non-increasing in $N$. To this end, we first show that $E^p(N,Z)$ can be computed varying only over fermionic density matrices with compactly supported integral kernel.
  \begin{lemma}
	Let $Z\geq0$, $N>0$ and $\gamma\in\mathcal{I}$ with $\spr{\gamma}=N$. Then, for any  $\epsilon>0$ there exists a  $\hat\gamma\in\mathcal{I}$ with a compactly supported integral kernel,  $\spr{
		\hat \gamma}=N$ and 
	\begin{align*}
	|\mathcal{E}_Z^p(\gamma) - \mathcal{E}_Z^p(\hat\gamma) |\leq \epsilon.
	\end{align*} 
	\begin{proof}
		Let $\gamma\in\mathcal{I}$ with $\spr{\gamma} =N$ be given. For $R>0$ define $\gamma_R:=\chi_R \gamma \chi_R$,
		where $\chi_R(x) = \chi\left(\frac{|x|}{R}\right)$ and $\chi:[0,\infty)\to[0,1]$ fulfills the following properties:
		\begin{itemize}
			\item{$\chi$ is non-increasing and smooth}
			\item{$\chi(x) = 1$ for $x\leq 1$}
			\item{$\chi(x) = 0$ for $x\geq 2$}
		\end{itemize}
	Hence, $0\leq \gamma_R\leq \gamma \leq 1$ and  $\spr{\gamma_R}\leq \spr{\gamma}$. To guarantee the correct normalization, we define
		\begin{align*}
		\hat\gamma_R (x,y) = \gamma_R(x,y) + c_R \chi_R(x-v)\chi_R(y-v)=:\gamma_R(x,y) + \delta_R(x,y),
		\end{align*}
		where $v\in \mathbb{R}^3$ is arbitrary, but fixed with $|v|>5R$. Also define $c_R $ by
		\begin{align}c_R:= \frac{\spr\gamma - \spr {\gamma_R}}{\int \chi_R(x)^2\dd{x}}.\label{defcr}
		\end{align}
	By construction, we have  $\mathrm{tr}{\hat\gamma_R}= \mathrm{tr}{\gamma}$ and $0\leq \gamma_R\leq 1$ for sufficiently large $R>0$. 
		
	First, we prove that $\gamma_R\to \gamma$ in the space of trace class operators $\mathfrak{S}_1$. It suffices to show that for $R\to\infty$,  $\|\gamma_R\|_{\mathfrak{S}_1} \to \|\gamma\|_{\mathfrak{S}_1}$ and $\gamma_R\to \gamma$ in the weak operator topology. \cite[page\ 47]{arazy1981more}.
		
		Convergence in the weak operator topology follows by Lebesgue's theorem, the pointwise convergence $\lim_{R\to\infty}\chi_R(x) \to 1$ and the fact that \begin{align*}\iint |\gamma(x,y) \psi(x) \phi(y)| \dl \dd{y} \leq \|\gamma\|_{HS}  \|\psi\|_{L^2} \|\phi\|_{L^2} \leq  \sqrt{ \mathrm{tr}\gamma} \|\psi\|_{L^2} \|\phi\|_{L^2}\end{align*} for any $\phi,\psi \in L^2(\mathbb{R}^3)$.
	
Convergence of norms also follows directly using Lebesgue's theorem and the convergence in weak operator sense.
		\begin{align*}
\lim_{R\to\infty} \|\gamma_R\|_{\mathfrak{S}_1}&=\lim_{R\to\infty}\mathrm{tr}{\gamma_R}=\lim_{R\to\infty} \sum_{i\in\mathbb{N}} \langle\phi_i, \gamma_R\phi_i\rangle \\ &= \sum_{i\in\mathbb{N}} \lim_{R\to\infty} \langle \phi_i,\gamma_R\phi_i\rangle = \sum_{i\in\mathbb{N}} \langle\phi_i,\gamma\phi_i\rangle = \spr{\gamma}=\|\gamma\|_{\mathfrak{S}_1}.
		\end{align*}
		
		Hence, $\gamma_R\to\gamma$ in $\mathfrak{S}_1$ as $R\to\infty$. To show $\hat{\gamma}_R \to \gamma$, it suffices to prove $\|\delta_R\|_{\mathfrak{S}_1}\xrightarrow{R\to\infty} 0$. This holds true, since 
		\begin{align*}
		\|\delta_R\|_{\mathfrak{S}_1} = c_R \int \chi_R(x-v) \chi_R(x-v) \dl = c_R \int \chi_R(x)^2 \dl = \mathrm{tr}{\gamma}-\mathrm{tr}{\gamma_R}
		\end{align*}
		and therefore, $\|\delta_R\|_{\mathfrak{S}_1} \to 0$.
		
		Now, we are in the position to prove $\mathcal{E}_Z^p(\hat\gamma_R) \xrightarrow{R\to\infty} \mathcal{E}_Z^p(\gamma)$. We start with the kinetic term
		\begin{align*}
		\mathrm{tr}({-\Delta\hat\gamma_R}) = \mathrm{tr}({-\Delta\gamma_R  -\Delta\delta_R}) = \mathrm{tr}{(-\Delta\chi_R\gamma\chi_R - \Delta\delta_R)}.
		\end{align*}
		For the first term $\mathrm{tr}(-\Delta\chi_R\gamma\chi_R)=\|\nabla\chi_R\gamma^{\frac{1}{2}}\|_{HS}^2$, by Lebesgue's theorem,  we have
		\begin{align*}
		\|\nabla\chi_R\gamma^{\frac{1}{2}}\|_{HS}^2 & =\iint |\gamma^{\frac 12}(x,y) \nabla\chi_R (x) + \chi_R(x) \nabla_x\gamma^{\frac 12}(x,y)|^2 \dd{x}\dd{y}\nonumber \\ &= \iint |\gamma^{\frac 12}(x,y)R^{-1} \nabla\chi (x) + \chi_R(x) \nabla_x\gamma^{\frac 12}(x,y)|^2 \dd{x}\dd{y}\nonumber \\ & \xrightarrow{R\to\infty} \iint |\nabla_x\gamma^{\frac 12} (x,y)|^2 \dd{x}\dd y = \mathrm{tr}({-\Delta \gamma}).
		\end{align*}
	The second term can be computed as follows
		\begin{align*}
		\spr{-\Delta\delta_R} = c_R \int \chi_R(x-v) -\Delta\chi_R(x-v) \dl = c_R R \int |\nabla \chi|^2(x) \dl\xrightarrow{R\to\infty} 0,\label{432}
		\end{align*}
		where we used that $c_R = \mathcal{O}(R^{-3})$.
		
		Using analogous arguments the convergence of the Coulomb term and the Hartree energy is straightforward to check. We will omit this here and finish by proving the convergence of the exchange term.		
		We will use Hardy's inequality $-\Delta_x \geq \frac{1}{4}|x-y|^{-2}$ to get
		\begin{align}\nonumber
		|X(\gamma^p) - X(\hat\gamma_R^p) |& \leq  \iint\left|\frac{|\gamma^p(x,y)|^2-|\hat\gamma_R^p(x,y)|^2}{2|x-y|}\right|\dl\dd{y} \\ & \leq \iint \frac{\left|\gamma^p(x,y) - \hat\gamma_R^p(x,y)\right|(|\gamma^p(x,y)|+|\hat\gamma_R^p(x,y)|)}{2|x-y|}\dl\dd{y} \nonumber
		\\ & \leq\left( \iint |\gamma^p(x,y) - \hat\gamma_R^p(x,y)|^2 \dd{x}\dd{y} \iint \frac{|\gamma^p(x,y)|^2+|\hat\gamma_R(x,y)|^2}{4|x-y|^2}\dl\dd{y}\right)^{1/2}\nonumber \\ & = \|\gamma^p-\hat\gamma_R^p\|_{\mathfrak{S}_2}\left( \iint |\nabla_x \gamma^p(x,y)|^2 + |\nabla_x \hat\gamma^p_R(x,y)|^2\dl\dd{y}\right)^{1/2}\nonumber \\ & =\nonumber \|\gamma^p-\hat\gamma_R^p\|_{\mathfrak{S}_2} \left( \spr{-\Delta\gamma^{2p}} + \spr{-\Delta\hat\gamma_R^{2p}} \right)^{1/2} \\ & \leq \|\gamma^p-\hat\gamma_R^p\|_{\mathfrak{S}_2} \left( \spr{-\Delta\gamma} + \spr{-\Delta\hat\gamma_R}\right)^{1/2}.
		\end{align}
		Since $\mathrm{tr}({-\Delta\hat\gamma_R)}\xrightarrow{R\to\infty}\mathrm{tr}({-\Delta\gamma})$, it suffices to show that $\hat\gamma_R^p \to \gamma^p$ in the Hilbert-Schmidt norm. We have shown that $\hat\gamma_R\to\gamma$ in $\mathfrak{S}_1$  and by the continuity of the embedding $\mathfrak{S}_1 \hookrightarrow \mathfrak{S}_{2p}$, we deduce that $\hat\gamma_R \to \gamma$ in $\mathfrak{S}_{2p}$. Moreover, the map $A\mapsto A^p, \mathfrak{S}_{2p} \to \mathfrak{S}_{2p} $ is continuous for $A\geq 0$ \cite[page 28]{simon2010trace}. This implies $\hat\gamma_R^{p}\to\gamma^p$ in the Hilbert-Schmidt norm and concludes the proof.
	\end{proof}
	\label{lem:compact}
\end{lemma}
\begin{lemma}[$N\mapsto E^p(N,Z)$ is non-increasing]\label{lem:non-inc}Given $N>0,Z\geq 0$ and any $M>0$. Then,
 \begin{align}
E^p(N+M,Z) \leq E^p(N,Z) + E^p(M,0)\leq E^p(N,Z).\label{eq:non-inc}
\end{align}
Note that $N\mapsto E^p(N,Z)$ is decreasing for $1/2\leq p <1$ since $E^p(M,0)<0$ in these cases.
\begin{proof}It suffices to prove the first inequality. For a contradiction assume that there exists a $\delta >0$ such that \begin{align}E^p(N+M,Z)> E^p(N,Z)+E^p(M,0)+\delta\label{eq:tocontradict}\end{align} for some $N>0$, $Z\geq 0 $ and $M>0$. Then, there exists a density matrix $\gamma_N$ with $\spr{ \gamma_N} = N$ such that $E^p(N,Z) > \mathcal{E}^p_Z(\gamma_N) - \delta/3$. By Lemma~\ref{lem:compact} we can assume that $\gamma_N(x,y)$ has compact support. Choose also a density matrix $\gamma_M$ with $\spr{\gamma_M}=M$ such that $\mathcal{E}^p_{Z=0}(\gamma_M) < E^p(M,0) +{\delta}/{3}$. We can again assume without loss of generality that $\gamma_M(x,y)$ is compactly supported.
	Denote by $R>0$ the radius of a ball in $\mathbb{R}^6$ which contains the supports of $\gamma_N(x,y)$ and $\gamma_M(x,y)$. 
	
	For an $\epsilon>0$, we define a translated operator by \begin{align*}
	\tilde{\gamma}_M (x,y):= \gamma_M(x-c,y-c)
	\end{align*}
	for a fixed $c\in\mathbb{R}^3$ satisfying $|c|> 2R + \frac 1 \epsilon$. Moreover, we define a trial density operator $\gamma_{N+M}$ to be $\gamma_{N+M}:= \gamma_N + \tilde \gamma_M$. By construction, we have 
	\begin{align*}
&\gamma_N \tilde \gamma_M = \tilde \gamma_M \gamma_N = 0,\;0 \leq \gamma_{N+M}\leq 1, \\ &\spr{\gamma_{N+M}} = N+M \text{ and }X(\gamma_{N+M}^p )= X(\gamma_N^p) + X(\tilde \gamma_M^p).
	\end{align*}
	Furthermore, since $\gamma_N\leq \gamma_{N+M}$, \begin{align*}
	\spr{ - \frac{Z}{|x|}\gamma_{N+M}} \leq 	\spr{ - \frac{Z}{|x|}\gamma_{N}}.
	\end{align*}
	For the Hartree term it is easy to see that
	\begin{align*}
	\iint \frac{ \gamma_{N}(x,x) \tilde \gamma_M(y,y)}{2|x-y|} \dl\dd{y} \leq \frac{\epsilon}{2} \iint  \gamma_{N}(x,x) \tilde \gamma_M(y,y) \dl \dd{y} = \frac{\epsilon N M }{2}
	\end{align*}
	and analogously\begin{align*}
	\iint \frac{ \gamma_{\mathrm{M}}(x,x) \tilde \gamma_N(y,y)}{2|x-y|} \dl\dd{y} \leq \frac{\epsilon}{2} \iint  \gamma_{\mathrm{M}}(x,x) \tilde \gamma_N(y,y) \dl \dd{y} = \frac{\epsilon N M }{2}.
	\end{align*}
Inserting everything into \eqref{eq:tocontradict} yields\begin{align*}
&\frac{2\delta}{3} + \mathcal{E}^p_Z(\gamma_N) + E^p(M,0 )< E^p(N+M,Z) \leq \mathcal{E}^p_Z(\gamma_{N+M})  \\&\leq \spr{-\Delta \gamma_N}+ \spr{-\Delta \tilde \gamma_M} + \spr{-\frac{Z}{|x|}\gamma_N} + D[\rho_{\gamma_N}] + D[\rho_{\tilde \gamma_{\mathrm{M}}}] + \epsilon NM - X(\gamma_N^p) - X(\tilde \gamma_M^p) \\ &= \mathcal{E}^p_Z(\gamma_N)+ \mathcal{E}^p_{Z=0}(\gamma_M) + \epsilon NM \\ &\leq \mathcal{E}^p_Z(\gamma_N)+ E^p(M,0) + \epsilon NM + \frac \delta 3,
\end{align*}
where we have used the translation invariance of $\mathcal{E}^p_{Z=0}$. Choosing $\epsilon = \delta/(3NM)$  gives a contradiction.
\end{proof}
\end{lemma}
This directly implies the following binding inequality for the minimizer.
\begin{cor}[Binding inequality]\label{cor:binding}For any smooth partition of unity $\chi_1^2 + \chi_2^2 =1$ we have \begin{align*}
	\mathcal{E}_Z^p(\gamma_0)\leq \mathcal{E}_Z^p(\chi_1\gamma_0\chi_1) + \mathcal{E}_{Z=0}^p(\chi_2 \gamma_0 \chi_2).
	\end{align*}
\end{cor}
\subsection{Localizing density matrices}
\begin{lemma}\label{lem}Let $\gamma\in\mathcal{I}$ and $0\leq \chi(x) \leq 1$ be a smooth function on $\mathbb{R}^3$. Then we have for all $p\in[ 1/2, 1]$
\begin{align}\label{eq:chip}
X(\chi \gamma^p\chi) \leq X((\chi \gamma\chi)^p)
\end{align}
and 
\begin{align}\label{eq:chip2}
X(\chi\gamma^p) \leq \left(\mathrm{tr}(-\Delta\chi\gamma\chi)\right)^{\frac 12} \left(\int \chi^2 \rho_\gamma\right)^{\frac 12}.
\end{align}
\label{pabschaetzung}
\begin{proof} We first prove \eqref{eq:chip2}. This is obtained using the Cauchy-Schwarz inequality, the Hardy inequality and the fact that $\gamma^{2p} \leq \gamma$.
\begin{align*}
X(\chi \gamma^p)& = \frac 12 \iint \frac{|\chi(x) \gamma^p(x,y)|^2}{|x-y|} \dl \dd{y} 
 \\& \leq  \left(\iint \frac{|\chi(x) \gamma^p(x,y)|^2}{4|x-y|^2} \dl\dd{y}\right)^{\frac 12} \left( \iint |\chi(x) \gamma^p(x,y)|^2 \dl \dd{y} \right)^{\frac 12} \\ 
  &\leq \left(\mathrm{tr} ( -\Delta\chi \gamma^{2p} \chi ) \right)^{\frac 12} \left(\int \chi^2 \rho_\gamma\right)^{\frac 12} \leq \left(\mathrm{tr}(-\Delta\chi\gamma\chi)\right)^{\frac 12} \left(\int \chi^2 \rho_\gamma\right)^{\frac 12}.
\end{align*}

Now, we prove \eqref{eq:chip}.
Using \eqref{eq:fefferman} it suffices to prove that \begin{align}
\chi \gamma^p \chi \leq (\chi \gamma \chi)^p.\label{515}\end{align}
 For $p=1$ the inequality is trivial and for $p=\frac{1}{2}$ it can be shown as follows:
\begin{align*}
\chi\gamma^{\frac 12} \chi = (\chi \gamma^{\frac 12} {\chi \chi} \gamma^{\frac{1}{2}}\chi)^{\frac{1}{2}} \leq (\chi \gamma\chi)^{\frac 12}.
\end{align*}
Now, we are left with the case $p\in( 1/2 , 1)$. Setting $\eta:=\gamma^p$, we can write inequality~\eqref{515} as
\begin{align*}
\chi\eta\chi \leq (\chi\eta^{\frac 1p}\chi)^p.
\end{align*} Since $ 1/2  < p < 1$, it is enough to show that
\begin{align}\label{518}
(\chi \eta \chi )^{ 1/p} \leq \chi \eta^{ 1/p}\chi .
\end{align}We note that
\begin{align}\label{integraldarstellung}
C^{ 1/p} = c_p \int_0^\infty C^2(C+z\mathbb{1})^{-1} z^{ \frac 1p -2} \dd{z}
\end{align}
for any non-negative self-adjoint bounded operator $C$ and some constant $c_p >0$.
Hence, inequality~\eqref{518} holds true, if
\begin{align*}
c_p \int_0^\infty \left[(\chi\eta\chi)^2(\chi\eta\chi+z\mathbb{1})^{-1} - \chi \eta^2(\eta+z\mathbb{1})^{-1}\chi \right] z^{\frac 1p -2} \dd{z} \leq 0 .
\end{align*}
Thus, it suffices to show that
\begin{align}
\chi(\chi\eta\chi+z\mathbb{1})^{-1}\chi \leq (\eta +  z\mathbb{1})^{-1} \label{zuzeigen}
\end{align} for all $z>0$. 

Note that for self-adjoint bounded operators $A,B$ with \begin{align}\label{8}A\geq B > 0 \text{ we have, that }B^{-1} \geq A^{-1}.\end{align}

To use this fact, we approximate $\chi$ with an invertible operator $\chi_\epsilon$. For any $0<\epsilon<1$, we define
\begin{align*}\chi_\epsilon({z})=\max(\epsilon,\chi({z})).\end{align*}
Obviously, $\chi_\epsilon \to \chi$ in norm as $\epsilon\to 0$. In particular, this implies 
\begin{align*}\chi_\epsilon (\chi_\epsilon \eta \chi_\epsilon + z\mathbb{1})^{-1} \chi_\epsilon \to \chi(\chi\eta\chi + z\mathbb{1})^{-1} \chi\end{align*} for all $z>0$.

Since $0< \chi_\epsilon \leq 1$ and $z>0$, we have
\begin{align*}
\eta + z\mathbb{1} \chi_\epsilon^{-2} \geq \eta + z\mathbb{1} .
\end{align*}
Using \eqref{8} it follows that
\begin{align*}
\chi_\epsilon (\chi_\epsilon \eta \chi_\epsilon + z\mathbb{1})^{-1} \chi_\epsilon \leq (\eta	+ z\mathbb{1})^{-1}.
\end{align*}
The limit $\epsilon\to 0$ shows~\eqref{zuzeigen}, which concludes the proof.
\end{proof}
\end{lemma}
\begin{cor}\label{cor:apriori}Let $\gamma_0$ be a minimizer of $E^p(N,Z)$. Then, 
\begin{align}\label{eq:corrho}
\int \rho_0^{\frac 53} + \mathrm{tr}(-\Delta \gamma_0) + D[\rho_{\gamma_0}] \leq C (Z^{\frac 73} + N)
\end{align}
and 
\begin{align}\label{eq:Xest}
X(\gamma_0^p) \leq C (Z^{\frac 73} + N)^{\frac 12} N^{\frac 12} .
\end{align}
\begin{proof}
From \eqref{eq:chip2} we know that 
\begin{align}\label{eq:chip3}
X(\gamma_0^p) \leq\mathrm{tr}(-\Delta \gamma_0)^{\frac 12} N^{\frac 12} .
\end{align}
Using this, the kinetic Lieb-Thirring inequality and the fact that the ground state energy in Thomas-Fermi theory equals a negative constant times $Z^{\frac 73}$, we estimate
\begin{align*}
\mathcal{E}^p_Z(\gamma_0) &\geq \frac 14 \mathrm{tr}(-\Delta \gamma_0 ) + \tilde C \int \rho_0^{\frac 53} - Z \int \frac{\rho_0(x)}{|x|} + D[\rho_0] - CN \\ 
&\geq \frac 14 \mathrm{tr}(-\Delta \gamma_0 ) + \frac{\tilde C}{2} \int \rho_0^{\frac 53}  + \frac 12 D[\rho_0] - C Z^{\frac 73}- CN .
\end{align*} 
The fact that $\mathcal{E}^p_Z(\gamma_0) \leq 0$ implies \eqref{eq:corrho}. This also shows $\mathrm{tr}(-\Delta \gamma_0) \leq C(Z^{\frac 73} + N)$ which proves \eqref{eq:Xest} using \eqref{eq:chip3}.
\end{proof}
\end{cor}
\begin{lemma}[IMS-type formula]
For all quadratic partitions of unity $\sum_{i=1}^n f_i^2 = 1 $ with $\nabla f_i \in L^\infty$ and for all density matrices $0\leq \gamma\leq 1$ with $\mathrm{tr} ( (1-\Delta) \gamma)<\infty$, we have \begin{align}\nonumber
\sum_{i=1}^n \mathcal{E}^p_Z ( f_i \gamma f_i) - \mathcal{E}^p_Z(\gamma) \leq& \int \left( \sum_{i=1}^{n}|\nabla f_i(x) |^2 \right) \rho_\gamma(x) \dl \\& + \sum_{i<j}^n \iint \frac{f_i(x)^2 \left( |\gamma^p(x,y)|^2 - \rho_\gamma(x) \rho_\gamma(y)\right) f_j(y)^2}{|x-y|} \dl\dd{y}.\label{eq:IMS}
\end{align}
\begin{proof}
	We estimate the kinetic term using the IMS-formula to obtain\begin{align*}
	\sum_{i=1}^{n}\mathrm{tr}(-\Delta f_i \gamma f_i) - \mathrm{tr}(-\Delta \gamma) = \mathrm{tr}\left( \left(\sum_{i=1}^n |\nabla f_i|^2 \right) \gamma \right) = \int \left( \sum_{i=1}^{n}|\nabla f_i|^2 \right) \rho_\gamma.
	\end{align*}
	For the direct term we have 
	\begin{align*}
\sum_{i=1}^{n}D[\rho_{f_i\gamma f_i}] = \sum_{i=1}^n D[f_i^2 \rho_\gamma] = D[\rho_\gamma] - \sum_{i<j}^{n} \iint \frac{f_i(x)^2 \rho_\gamma(x) \rho_\gamma(y) f_j(y)^2}{|x-y|} \dl\dd{y}.
	\end{align*}
	The exchange term can be estimated using \eqref{eq:chip} to get \begin{align*}
	-\sum_{i=1}^n &X((f_i  \gamma f_i )^p) \leq - \sum_{i=1}^{n} X(f_i \gamma^p f_i ) \\ & = -X(\gamma^p) + \sum_{i<j}^n \iint \frac{f_i^2(x) |\gamma^p(x,y)|^2 f_j^2(y)}{|x-y|} \dl\dd{y}.
	\end{align*}
\end{proof}
\end{lemma}

The rest of the paper will be completely analogous to the corresponding parts in \cite{frank2016maximal} and \cite{frank2016ionization}. For convenience of the reader it is included here as well.

\section{\texorpdfstring{Exterior $L^1$-estimate}{Exterior L1-estimate}}\label{sec:l1est}
In this section we control $\int_{|x|>r} \rho_0$. First, we recall the screened nuclear potential \begin{align*}
\Phi_r (x) = \frac{Z}{|x|} - \int_{|y|<r} \frac{\rho_0(y)}{|x-y|}\dd{y}.
\end{align*}
We also introduce the cut-off function
\begin{align*}
\chi_r^+ = \mathbb{1}(|x|>r)
\end{align*} 
and a smooth function $\eta_r: \mathbb{R}^3\to [0,1]$ satisfying 
\begin{align}
\chi_r^+ \geq \eta_r \geq \chi_{(1+\lambda)r}^+, \; |\nabla \eta_r| \leq C (\lambda r)^{-1}
\end{align}
for some $\lambda >0$.
\begin{lemma}\label{lem:keylemmaext}
	For all $r>0, s>0, \lambda\in (0,1/2]$ we have \begin{align*}
	\int \chi_r^+ \rho_0 \leq C \int_{r\leq |x| < (1+\lambda)^2 r} \rho_0 + C \left( \sup_{|z|>r} [ |z|\Phi_r(z)]_+ + s + \lambda^{-2} s^{-1} + \lambda^{-1} \right) \\ 
+	C \left(s^2 \mathrm{tr}(-\Delta \eta_r \gamma_0 \eta_r) \right)^{ 3/5} + C \left( s^2 \mathrm{tr}(-\Delta \eta_r\gamma_0 \eta_r) \right)^{ 1/3}.
	\end{align*}
	\begin{proof}
Recall from Corollary~\ref{cor:binding} that the minimizer $\gamma_0$ fulfills the binding inequality
		\begin{align}\label{eq:binding}
		\mathcal{E}^p_Z(\gamma)\leq \mathcal{E}^p_Z(\chi_1\gamma\chi_1) + \mathcal{E}^p_{Z=0}(\chi_2\gamma\chi_2)
		\end{align}
		for any smooth partition of unity $\chi_1^2 + \chi_2^2=1$.
						
		 For fixed $\lambda\in (0,1/2], s>0, l>0, \nu \in \mathbb{S}^2$ we choose
		\begin{align*}
		\chi_1(x) = g_1\left(\frac{\nu\cdot \theta (x)-l}{s}\right), \chi_2(x)= g_2\left(\frac{\nu\cdot \theta (x)-l}{s}\right),
		\end{align*}
	where $g_1,g_2: \mathbb{R}\to\mathbb{R}$ and $\theta: \mathbb{R}^3\to\mathbb{R}^3$ satisfy 
	\begin{align*}&g_1^2+g_2^2 = 1, \; g_1(t) = 1 \text{ if }t\leq 0,\;  g_1(t) =0 \text{ if }t\geq 1,\;  |g_1^\prime |+ |g_2^\prime|\leq C, \\
	&|\theta(x)|\leq |x|, \theta(x) = 0 \text{ if } |x|\leq r, \; \theta(x) = x \text{ if } |x| \geq (1+\lambda)r \text{ and }|\nabla \theta|\leq C\lambda^{-1}.
\end{align*}
Now, we begin to estimate the binding inequality \eqref{eq:binding} using the IMS-type formula \eqref{eq:IMS}.
\begin{align*}
\mathcal{E}^p_Z &(\chi_1 \gamma_0 \chi_1) + \mathcal{E}^p_Z(\chi_2 \gamma_0 \chi_2) - \mathcal{E}^p_Z (\gamma_0) \\ \leq &\int \left( |\nabla \chi_1|^2 + |\nabla \chi_2|^2 \right) \rho_0 + \int \frac{Z \chi_2^2 \rho_0(x)}{|x|}\dl \\& + \iint \frac{\chi_2^2 (x) \left( |\gamma_0^p(x,y)|^2 - \rho_0(x) \rho_0(y) \right) \chi_1^2(y)}{|x-y|}\dl \dd{y}.
\end{align*}
By construction we have \begin{align*}
\int \left( |\nabla \chi_1|^2 + |\nabla \chi_2|^2\right) \rho_0 \leq C (1+(\lambda s)^{-2})\int_{\nu\cdot  \theta(x) -s \leq l \leq \nu \cdot \theta(x)} \rho_0(x) \dl.
\end{align*}
For the attraction and direct terms, we can estimate
\begin{align*}
\int \frac{Z \chi_2^2(x) \rho_0(x)}{|x|}&\dl - \iint \frac{\chi_2^2(x) \rho_0 (x) \chi_1^2(y)\rho_0(y)}{|x-y|}\dl \dd{y} \\& = \int \chi_2^2 (x) \rho_0(x) \Phi_r(x) \dl - \iint_{|y|\geq r} \frac{\chi_2^2(x) \rho_0(x) \chi_1^2(y) \rho_0(y)}{|x-y|}\dl \dd{y}\\ & \leq \int_{l\leq \nu\cdot \theta(x) } \rho_0(x) [\Phi_r(x)]_+ \dl - \iint_{|y|\geq r, \nu\cdot \theta(y)\leq l \leq \nu\cdot \theta(x) -s } \frac{\rho_0(x) \rho_0(y)}{|x-y|}\dl\dd{y}.
\end{align*}
Since $\theta(x) = x$ when $|x| \geq (1+\lambda)r$, we obtain\begin{align*}
\iint _{|y|>r, \nu\cdot \theta(y) \leq \nu \cdot \theta(x) -s }\frac{\rho_0(x) \rho_0(y)}{|x-y|}\dl\dd{y} \geq \iint_{\stackrel{|x|,|y|\geq(1+\lambda)r}{ \nu\cdot y \leq l\leq  \nu\cdot x -s}} \frac{\rho_0(x)\rho_0(y)}{|x-y|}\dl\dd{y}.
\end{align*}
For the exchange-correlation term, we use
\begin{align}
\iint \frac{\chi_2^2(x) |\gamma_0^p(x,y)|^2 \chi_1^2(y)}{|x-y|} \dl\dd{y} \leq \iint_{\nu\cdot \theta(y) -s \leq l \leq \nu \cdot \theta(x) }\frac{|\gamma_0^p(x,y)|^2}{|x-y|}\dl\dd{y}.
\end{align}
Now we apply these results to the binding inequality \eqref{eq:binding} to obtain
\begin{align}\nonumber
&\iint_{\stackrel{|x|,|y|\geq(1+\lambda)r}{ \nu\cdot y \leq l\leq \nu\cdot x -s}} \frac{\rho_0(x) \rho_0(y)}{|x-y|}\dl\dd{y} \leq C (1+(\lambda s)^{-2})\int_{\nu \cdot \theta(x) -s \leq l \leq \nu \cdot \theta(x)} \rho_0(x) \dl \\
& + \int_{l\leq \nu\cdot \theta(x) } \rho_0(x) [\Phi_r(x)]_+ \dl + \iint_{\nu\cdot \theta(y) -s \leq l \leq \nu \cdot \theta(x) }\frac{|\gamma_0^p(x,y)|^2}{|x-y|}\dl\dd{y}\label{eq:tointegrate}
\end{align}
for all $s>0, l>0$ and $\nu \in \mathbb{S}^2$. Now we want to integrate \eqref{eq:tointegrate} over $l\in (0,\infty)$ and then average over $\nu\in\mathbb{S}^2$. To do so, we first write the left side as follows.
\begin{align*}
\int_{\mathbb{S}^2} \frac{\dd\nu}{4\pi } &\int_{0}^{\infty}\dd l \iint_{\stackrel{|x|,|y|\geq(1+\lambda)r}{ \nu\cdot y \leq l\leq \nu\cdot x -s}} \frac{\rho_0(x) \rho_0(y)}{|x-y|}\dl\dd{y}\\  = &\frac 12 \int_{\mathbb{S}^2} \frac{\dd\nu}{4\pi } \int_{0}^{\infty}\dd l \iint_{\stackrel{|x|,|y|\geq(1+\lambda)r}{ \nu\cdot y \leq l\leq \nu\cdot x -s} } \frac{\rho_0(x) \rho_0(y)}{|x-y|}\dl\dd{y}\\ & + \frac 12 \int_{\mathbb{S}^2} \frac{\dd\nu}{4\pi } \int_{0}^{\infty}\dd l \iint_{\stackrel{|x|,|y|\geq(1+\lambda)r}{ -\nu\cdot x \leq l\leq -\nu\cdot y -s} } \frac{\rho_0(x) \rho_0(y)}{|x-y|}\dl\dd{y}
\end{align*}
In the second term we used the symmetries $\nu\mapsto -\nu$ and $x\leftrightarrow y$.
For $a:=\nu\cdot x$ and $b:=\nu\cdot y$ remark that
\begin{align*}
\int_0^\infty \mathbb{1}(b\leq l\leq  a -s) + \mathbb{1}(-a \leq l \leq -b-s) \dd{l}&   \geq \left[[a-b]_+ -2s\right]_+ \geq [a-b]_+ -2s.
\end{align*}
Also note that \begin{align*}
\int_{\mathbb{S}^2} [\nu \cdot z ]_+ \frac{\dd{\nu}}{4\pi} = \frac{|z|}{4}
\end{align*}
for any $z\in \mathbb{R}^3$ and 
\begin{align*}
\int_0^\infty \mathbb{1}(b-s \leq l \leq a) \dd{l} \leq [a-b]_+ + s.
\end{align*}
We will also use Fubini's theorem to interchange the integrals. For the right hand side, we use the fact that 
\begin{align*}
\{ x: \nu\cdot \theta(x)\geq l  \}\subset \{ x: |x| \geq r \}
\end{align*}
since $l>0$ and $\theta(x)=0$ when $|x|<r$. Thus, after integrating $l$ from $0$ to $\infty$ and averaging over $\nu \in\mathbb{S}^2$, inequality~\eqref{eq:tointegrate} gives 
\begin{align*}
\frac 12 &\iint_{|x|,|y|\geq (1+\lambda)r}\frac{|x-y|/4 -2s}{|x-y|}\rho_0(x) \rho_0(y) \dl\dd{y} \\ \leq & C (s+\lambda^{-2}s^{-1})\int_{|x|\geq r} \rho_0(x) \dl + \int_{|x|\geq r} [ |\theta(x) |/4 \Phi_r(x) ]_+  \rho_0(x) \dl \\ & + \iint_{|x|\geq r} \frac{|\theta(x) - \theta(y)|/4 +s}{|x-y|} |\gamma_0^p(x,y)|^2 \dl \dd{y}.
\end{align*}
Using $|\theta(x) |\leq |x|$ and $\theta(x) - \theta(y)| \leq C \lambda^{-1} |x-y|$, this simplifies to \begin{align*}
\frac{1}{8} \left( \int \chi_{(1+\lambda)r}^+ \rho_0\right)^2 \leq &  \left( \frac 14 \sup_{|z|\geq r} [|z| \Phi_r(z)]_+ + C s + C \lambda^{-2} s^{-1} + C \lambda^{-1} \right) \int \chi_r^+ \rho_0 \\ & + s D[\chi_{(1+\lambda)r}^+ \rho_0] + s \iint \frac{\chi_r^+(x) |\gamma^p(x,y)|^2}{|x-y|}\dl \dd y . 
\end{align*}
Now, we replace $r$ by $(1+\lambda)r$ to get
\begin{align}\nonumber
\frac{1}{8} & \left( \int \chi_{(1+\lambda)^2r}^+ \rho_0\right)^2  \leq  \\\nonumber  & \left( \frac 14 \sup_{|z|\geq (1+
	\lambda) r}  [|z| \Phi_{(1+\lambda)r}(z)]_+ + C s + C \lambda^{-2} s^{-1} + C \lambda^{-1} \right) \int \chi_{(1+
	\lambda)r}^+ \rho_0 \\ & \,\, + s D[\chi_{(1+\lambda)^2r}^+ \rho_0] + s \iint \frac{\chi_{(1+\lambda)r}^+(x) |\gamma^p(x,y)|^2}{|x-y|}\dl \dd y . \label{eq:L1est}
\end{align}
First, we estimate the left hand side of \eqref{eq:L1est}. \begin{align*}
	\left( \int \chi^+_{(1+\lambda)^2r} \rho_0 \right)^2 \geq \frac 12 \left( \int \chi_r^+ \rho_0 \right)^2 - \left( \int_{r<|x|<(1+\lambda)^2 r} \rho_0\right)^2.
\end{align*}
Now we also estimate the right side of \eqref{eq:L1est}. For the first term we use $\Phi_{(1+\lambda)r}(z) \leq \Phi_r(z)$ and $\chi_{(1+\lambda) r } \leq \chi_r$ to get \begin{align*}
&\left( \frac 14 \sup_{|z| \geq (1+\lambda) r} \left[ |z| \Phi_{(1+\lambda)r}(z)  \right]_+ + C s +  C \lambda^{-2}s^{-1} + C \lambda^{-1} \right) \int \chi_{(1+\lambda)r}^+ \rho_0 \\ 
&\leq \left( \frac 14 \sup_{|z|\geq r} [|z| \Phi_r(z)]_+ + C s + C \lambda^{-2} s^{-1} + C \lambda^{-1} \right) \int \chi_r^+ \rho_0.
\end{align*}
For the second term we use the Hardy-Littlewood-Sobolev, the H\"older and the Lieb-Thirring inequalities to obtain 
\begin{align*}
D[\chi^+_{(1+\lambda)^2 r} \rho_0]\leq C \| \chi_{(1+\lambda)r}\rho_0 \|_{L^{6/5}}^2 \leq C \| \chi_{(1+\lambda)r}\rho_0 \|_{L^1}^{7/6} \| \chi_{(1+\lambda)r}\rho_0 \|_{L^{5/3}}^{5/6} \\ \leq C \|\chi_r^+ \rho_0\|_{L^1}^{7/6} \left(\mathrm{tr}(-\Delta \eta_r \gamma_0 \eta_r) \right)^{\frac 12}.
\end{align*}
We also used $\eta_r^2 \geq \chi^+_{(1+\lambda)^2 r}$. For the third term we use \eqref{eq:chip2} to get 
\begin{align*}
\iint &\frac{\chi_{(1+\lambda)r}^+(x) |\gamma^p_0(x,y)|^2}{|x-y|} \dl \dd{y} \leq \iint \frac{\eta_r(x)^2 |\gamma^p_0(x,y)|^2}{|x-y|} \dl \dd{y}\\ &\leq 2 \left( \mathrm{tr}(-\Delta \eta_r \gamma_0 \eta_r) \right)^{1/2} \left( \int \chi_r^+ \rho_0 \right)^{1/2}.
\end{align*}
Putting all the estimates in \eqref{eq:L1est} we end up with \begin{align*}
\left( \int \chi_r^+ \rho_0\right)^2 \leq  & C \left( \int_{r<|x| < (1+\lambda)^2 r} \rho_0 \right)^2\\
& + C \left( \sup_{|z|\geq r} [ |z| \Phi_r (z) ]_+ + s + \lambda^{-2} s^{-1} + \lambda^{-1} \right) \int \chi_r^+ \rho_0 \\ & + C s \left( \int \chi_r^+ \rho_0 \right)^{\frac 76} \left(\mathrm{tr}(-\Delta \eta_r \gamma_0 \eta_r) \right)^{\frac 12} 
\\ & + C s \left( \mathrm{tr}(-\Delta \eta_r \gamma_0 \eta_r ) \right)^{\frac 12} \left( \int \chi_r^+ \rho_0\right)^{\frac 12}.
\end{align*}
Hence, by Young's inequality,
\begin{align*}
\int \chi_r^+ \rho_0 \leq C &\int_{r< |x| < (1+\lambda)^2 r} \rho_0 + C \left( \sup_{|z|\geq r} [|z| \Phi_r(z) ]_{+} + s + \lambda^{-2} s^{-1} + \lambda^{-1} \right)  \\ &+ C \left( s^2 \mathrm{tr}(-\Delta \eta_r \gamma_0 \eta_r ) \right)^{3/5} + C \left( s^2 \mathrm{tr}(-\Delta \eta_r \gamma_0 \eta_r) \right)^{1/3}.
\end{align*}
	\end{proof}\end{lemma}
From this proof we already get an upper bound on the excess charge.
\begin{cor}
	For the minimizer $\gamma_0$ we have \begin{align}\label{eq:aprioribound}
	\mathrm{tr}\gamma_0 = N \leq 2Z + C (Z^{\frac 23} +1).
	\end{align}
	Moreover, 
	\begin{align}\label{eq:aprioribound2}\int \rho_0^{\frac 53} + \mathrm{tr}(-\Delta \gamma_0) + D[\rho_0] \leq C( Z^{\frac 73} + 1)\end{align}
	and 
	\begin{align}\label{eq:aprioribound3}
	X(\gamma_0^p) \leq C (Z^{\frac 53} +1).
	\end{align}
	\begin{proof}
		Choosing $\lambda= \frac 12 $ and $r\to 0^+$ in \eqref{eq:L1est} leads to \begin{align*}
		N^2 \leq (2Z + Cs + Cs^{-1} C )N + C s D[\rho_0] + C s X(\gamma_0^p).
		\end{align*}
		Optimizing over $s>0$ we deduce that 
		\begin{align}\label{eq:l1est2}
		N \leq 2Z + C + C \left( (D[\rho_0] + X(\gamma_0^p) + N ) N^{-1} \right)^{\frac 12}.
		\end{align}
		Using the bounds from Corollary~\ref{cor:apriori} we get
		\begin{align*}
		D[\rho_0]+X(\gamma^p) \leq C (Z^{\frac 73} + N).
		\end{align*}
		Inserting this in \eqref{eq:l1est2} proves \eqref{eq:aprioribound}. Then, the bounds \eqref{eq:aprioribound2} and \eqref{eq:aprioribound3} follow from Corollary~\ref{cor:apriori}.
	\end{proof}
\end{cor}
\section{Splitting outside from inside}
In this section we want to estimate the difference of a reduced Hartree-Fock energy between the minimizer $\gamma_0$ and other density matrices away from the nucleus. 
The reduced Hartree-Fock functional is given by  \begin{align}
\mathcal{E}_r^{\mathrm{RHF}}(\gamma) = \mathrm{tr}(-\Delta \gamma ) - \int \Phi_r(x) \rho_\gamma(x) \dl + D[\rho_\gamma].
\end{align}
Note that this functional depends on the minimizer $\gamma_0$.
First, recall that we have introduced a smooth cut-off function $\eta_r:\mathbb{R}^3 \to [0,1]$ satisfying 
\begin{align}
\chi_r^+ \leq \eta_r \leq \chi_{(1+\lambda) r}^+
\end{align}
with $\lambda \in (0,1/2]$. Now we choose a quadratic partition of unity $\eta_-^2 +  \eta_{(0)}^2 + \eta_r^2 = 1$ with \begin{align*}
	\supp \eta_- \subset \{ |x|\leq r \}, \;\; \supp \eta_{(0)} \subset \{ (1-\lambda)r \leq |x| \leq (1+\lambda)r \}, \\
	\eta_-(x) =1 \text{ if } |x| \leq (1-\lambda) r, \;\; |\nabla \eta_-|^2 + |\nabla \eta_{(0)} |^2 + |\nabla \eta_r|^2 \leq C (\lambda r)^{-2}.
\end{align*} 
We will prove 
\begin{lemma}
For all $r>0$, all $\lambda \in (0,1/2]$, all density matrices $0\leq \gamma\leq 1$ satisfying $\supp \rho_\gamma \subset \{ x: |x|\geq r  \}$ and $\mathrm{tr} \gamma \leq \int \chi_r^+ \rho_0$ we have \begin{align*}
\mathcal{E}_r^{\mathrm{RHF}} (\eta_r \gamma_0 \eta_r )  \leq \mathcal{E}_r^{\mathrm{RHF}} (\gamma) + \mathcal{R},
\end{align*}
where \begin{align*}
\mathcal{R} \leq C (1+(\lambda r)^{-2} )  \int_{(1-\lambda)r \leq |x| \leq (1+\lambda) r} \rho_0 + C \lambda^3 r^3 \sup_{|z|\geq (1-\lambda) r}[\Phi_{(1-\lambda)r} (z) ]_+^{ 5/2} \\
+ C \left( \mathrm{tr}(-\Delta \eta_r \gamma_0 \eta_r ) \right)^{1/2} \left( \int \eta_r \rho_0 \right)^{1/2}
.\end{align*}
\begin{proof}
	It suffices to show that \begin{align}
	\mathcal{E}^p_Z(\eta_-\gamma_0\eta_-) + \mathcal{E}_r^{\mathrm{RHF}}(\eta_r \gamma_0 \eta_r) - \mathcal{R} \leq \mathcal{E}^p_Z(\gamma_0) \leq \mathcal{E}^p_Z(\eta_- \gamma_0 \eta_-) + \mathcal{E}^{\mathrm{RHF}}_r(\gamma).\label{eq:sufoutside}
	\end{align}
	
	\textbf{Upper bound.} Since $\gamma_0$ is a minimizer and by Lemma~\ref{lem:non-inc}, we have \begin{align}\label{eq:non-incupper}
	\mathcal{E}_Z^p(\gamma_0) \leq \mathcal{E}^p_Z(\eta_- \gamma_0 \eta_- + \gamma).
	\end{align}
	Since $\eta_-$ and $\rho_\gamma$ have disjoint supports, we have \begin{align*}
	(\eta_- \gamma_0 \eta_- + \gamma)^p = (\eta_- \gamma_0 \eta_-)^p + \gamma^p
	\end{align*}
	and 
	\begin{align*}
|(\eta_- \gamma_0 \eta_- + \gamma)^p(x,y)|^2 =   |(\eta_- \gamma_0 \eta_-)^p(x,y)|^2 + |\gamma^p(x,y|^2.
	\end{align*}
	Hence, \begin{align*}
	X( (\eta_- \gamma_0 \eta_- + \gamma)^p) = X( (\eta_- \gamma_0 \eta_-)^p)+ X(\gamma^p)
	\end{align*}
	and 
	\begin{align*}
	\mathcal{E}^p_Z(\eta_- \gamma_0 \eta_- + \gamma) &= \mathcal{E}^p_Z (\eta_- \gamma_0 \eta_-) + \mathcal{E}^p_Z(\gamma) + \iint \frac{\eta_-^2(x) \rho_0(x)\rho_\gamma(y)}{|x-y|}\dl \dd{y} \\ &\leq \mathcal{E}^p_Z(\eta_- \gamma_0 \eta_-) + \mathcal{E}^{\mathrm{RHF}}_{r=0}(\gamma) + \iint_{|x|\geq r} \frac{\rho_0(x) \rho_\gamma(y)}{|x-y|} \\ & 
= \mathcal{E}_Z^p(\eta_- \gamma_0 \eta_-)  + \mathcal{E}_r^{\mathrm{RHF}} (\gamma) .
	\end{align*}
	Inserting this in \eqref{eq:non-incupper} finishes the proof of the upper bound.
	
	\textbf{Lower bound.} Using the IMS-type formula \eqref{eq:IMS} and properties of the partition of unity, we have \begin{align*}
\mathcal{E}_Z^p(\gamma_0 )\geq & \mathcal{E}^p_Z(\eta_- \gamma_0 \eta_-) + \mathcal{E}^p_Z(\eta_{(0)}\gamma_0 \eta_{(0)} + \mathcal{E}_Z^p(\eta_r \gamma_0 \eta_r) \\ & - \int \left( |\nabla \eta_-|^2 + |\nabla \eta_{(0)}|^2 + |\nabla \eta_r|^2\right) \rho_0 \\ & + \iint \frac{\eta_r(x)^2 \rho_0(x) \rho_0(y) (\eta_-(y)^2 + \eta_0(y)^2 )}{|x-y|}\dl \dd{y} \\
& + \iint \frac{\eta_{(0)}^2(x) \rho_0(x) \rho_0(y) \eta_-(y)^2}{|x-y|}\dl\dd{y} \\ & - \iint \frac{(\eta_r^2(x) +\eta^2_{(0)}) |\gamma_0^p(x,y)|^2}{|x-y|}\dl \dd{y}
	\end{align*}
	and 
	\begin{align*}
	-\int \left(|\nabla \eta_-|^2 + |\nabla \eta_{(0)}|^2 + |\nabla \eta_r|^2 \right)\rho_0 \geq - C (\lambda r)^{-2} \int_{(1-\lambda) r \leq |x| \leq (1+\lambda)r} \rho_0 .
	\end{align*}
	Moreover, \begin{align*}
	\mathcal{E}^p_Z(\eta_r \gamma_0 \eta_r) &+ \iint \frac{\eta_r(x)^2 \rho_0(x) \rho_0(y) (\eta_-(y)^2 + \eta_{(0)}(y)^2) }{|x-y|}\dl \dd{y} - \iint \frac{ - \eta_r(x)^2 |\gamma_0^p(x,y)|^2}{|x-y|}\dl\dd{y} \\ &\geq \mathcal{E}^p(\eta_r \gamma_0 \eta_r) + \iint_{|y|\leq r}\frac{\eta_r(x)^2 \rho_0(x) \rho_0(y)}{|x-y|} - \iint \frac{ - \eta_r(x)^2 |\gamma_0^p(x,y)|^2}{|x-y|}\dl\dd{y}\\& = \mathcal{E}^{\mathrm{RHF}}_r(\eta_r \gamma_0 \eta_r) - X((\eta_r \gamma_0 \eta_r)^p)- \iint \frac{\eta_r(x)^2 |\gamma_0^p(x,y)|^2}{|x-y|} \dl \dd{y}\\ &
	\geq \mathcal{E}_r^{\mathrm{RHF}}(\eta_r \gamma_0 \eta_r) - 3 \left( \mathrm{tr}(-\Delta \eta_r \gamma_0 \eta_r)\right)^{1/2}\left( \int \eta_r^2 \rho_0\right)^2.
	\end{align*}
	We used \eqref{eq:chip2} twice, once with $\chi=1$ for $X((\eta_r\gamma_0 \eta_r)^p)$ and once with $\chi=\eta_r$ for $\iint \frac{\eta_r(x)^2 |\gamma_0^p(x,y)|^2}{|x-y|} \dl \dd{y}$. Similarly, we get
	\begin{align}\nonumber 
	\mathcal{E}^p_Z(\eta_{(0)}\gamma_0 \eta_{(0)}) &\nonumber + \iint \frac{\eta_{(0)}(x)^2 \rho_0(x) \rho_0(y) \eta_-(y)^2}{|x-y|} - \iint \frac{\eta_{(0)}(x)^2 |\gamma^p(x,y)|^2)}{|x-y|} \\ \geq & \nonumber	\mathcal{E}^p_Z(\eta_{(0)}\gamma_0 \eta_{(0)}) + \iint_{|y|\leq (1-\lambda )r} \frac{\eta_{(0)}(x)^2 \rho_0(x) \rho_0(y)}{|x-y|} \dl\dd{y}\\ &\nonumber  - \iint \frac{\eta_{(0)}(x)^2 |\gamma^p_0(x,y)|^2}{|x-y|}\dl \dd{y} \\ = &\nonumber \mathcal{E}^{\mathrm{RHF}}_{(1-\lambda)r}(\eta_{(0)} \gamma_0 \eta_{(0)})  - X((\eta_{(0)} \gamma_0 \eta_{(0)})^p) - \iint \frac{\eta_{(0)} |\gamma_0^p(x,y)|^2}{|x-y|}\dl \dd{y}\\ \geq & \nonumber \mathcal{E}^{\mathrm{RHF}}_{(1-\lambda)r}(\eta_{(0)} \gamma_0 \eta_{(0)}) - 3 \left( \mathrm{tr}(-\Delta  \eta_{(0)} \gamma_0 \eta_{(0)} )\right)^{1/2}\left( \int \rho_0 \eta_{(0)}^2\right)^{1/2}\\& \geq \mathrm{tr}\left( (-\frac 12 \Delta   - \Phi_{(1-\lambda)r}) \eta_{(0)} \gamma_0 \eta_{(0)} \right) - C \int \eta_{(0)}^2 \rho_0.\label{eq:lowerboundinside}
	\end{align}
	Again, we have used \eqref{eq:chip2}. Now, we apply the Lieb-Thirring inequality with $V=\Phi_{(1-\lambda)r} \mathbb{1}_{\supp \eta_{(0)}}$ to obtain
	\begin{align*}
	\mathrm{tr}\left( (-\frac 12 \Delta   - \Phi_{(1-\lambda)r}) \eta_{(0)} \gamma_0 \eta_{(0)} \right) \geq \mathrm{tr}([-\frac 12 \Delta - V]_-) \geq -C \int V^{\frac 52} \\ \geq -C\lambda r^3 \sup_{|x| \geq (1-\lambda)r} [ \Phi_{(1-\lambda)r}(x) ]_+^{5/2}.
	\end{align*}
	Plugging this estimate into \eqref{eq:lowerboundinside} yields
	\begin{align*}
\mathcal{E}^p_Z(\eta_{(0)}\gamma_0 \eta_{(0)}) &\nonumber + \iint \frac{\eta_{(0)}(x)^2 \rho_0(x) \rho_0(y) \eta_-(y)^2}{|x-y|} - \iint \frac{\eta_{(0)}(x)^2 |\gamma^p(x,y)|^2)}{|x-y|} \\ & \geq -C\lambda r^3 \sup_{|x| \geq (1-\lambda)r} [ \Phi_{(1-\lambda)r}(x) ]_+^{5/2} - C \int_{(1-\lambda)r \leq |x| \leq (1+\lambda)r} \rho_0.
	\end{align*}
	In total we get
	\begin{align*}
	\mathcal{E}^p_Z(\gamma_0) \geq \mathcal{E}^p(\eta_- \gamma_0 \eta_0) + \mathcal{E}^{\mathrm{RHF}}_r(\eta_r \gamma_0 \eta_r) - C(1+(\lambda r)^{-2}) \int_{(1-\lambda)r \leq |x| \leq (1+\lambda)r} \rho_0 \\ - C \lambda r^3 \sup_{|x| \geq (1-\lambda)r} [ \Phi_{(1-\lambda)r}(x)]_+^{5/2} - 3 \left(\mathrm{tr}(-\Delta \eta_r \gamma_0 \eta_r) \right)^{1/2} \left( \int \eta_r \rho_0 \right)^{1/2},
	\end{align*}
	which gives the lower bound in \eqref{eq:sufoutside}.
\end{proof}
\label{lem:insideout}
\end{lemma}
The previous lemma also implies
\begin{lemma}\label{lem:extkin}
	For all $r>0$ and  all $\lambda \in (0,1/2]$ we have \begin{align*}
	\mathrm{tr}(-\Delta \eta_r \gamma_0 \eta_r) \leq C (1+ (\lambda r)^{-2}) \int \chi_{(1-\lambda)r}^+ \rho_0 
	  + C \lambda r^3 \sup_{|z| \geq (1-\lambda)r} [\Phi_{(1-\lambda)r} (z) ]_{+}^{5/2}]
	\\ + C \sup_{|z| \geq r} [|z| \Phi_r(z)]_+^{7/3}.
	\end{align*}
	\begin{proof}
		Applying Lemma~\ref{lem:insideout} to $\gamma=0$ gives $\mathcal{E}_r^{\mathrm{RHF}}(\eta_r \gamma_0 \eta_r) \leq \mathcal{R}$. Using the Lieb-Thirring inequality and the fact that the ground state energy in Thomas-Fermi theory is a negative constant times $Z^{7/3}$, we can bound $\mathcal{E}_r^{\mathrm{RHF}}(\eta_r \gamma_0 \eta_r)$ from below:
		\begin{align*}
		\mathcal{E}_r^{\mathrm{RHF}}(\eta_r \gamma_0 \eta_r) &\geq \frac 12 \mathrm{tr}(-\Delta \eta_r \gamma_0 \eta_r) + C^{-1} \int (\eta_r \rho_0)^{5/3} - \sup_{|z| \geq r}[ |z| \Phi_r(z) ]_+ \int \frac{\eta_r^2 \rho_0}{|x|} + D[\eta_r \rho_0] \\& = \frac 12 \mathrm{tr}(-\Delta \eta_r \gamma_0 \eta_r) - C \sup_{|z|\geq r} [ |z| \Phi_r(z)]_+^{7/3}.
		\end{align*}
		Hence,
		\begin{align*}
		\mathrm{tr}(-\Delta \eta_r \gamma_0 \eta_r) \leq  C \mathcal{R} + C \sup_{|z|\geq r}[ |z| \Phi_r(z)]_+^{7/3},
		\end{align*}
		which implies the lemma.
	\end{proof}
\end{lemma}
\section{A collection of useful facts}
\subsection{Semiclassical analysis} In order to compare $\mathcal{E}^p_Z$ with Thomas-Fermi theory, we use a semiclassical approximation. The following results are taken from \cite[Lemma~8.2]{solovej2003ionization} after optimising over $\delta >0$ and replacing $V$ by $2V$. Moreover, $L_{\mathrm{sc}}=(15 \pi^2)^{-1}$.
\begin{lemma}\label{lem:semicl}
	For $s>0$, fix a smooth function $g: \mathbb{R}^3 
	\to [0,1]$ such that 
	\begin{align}
	\mathrm{supp} g\subset \{|x|\leq s\}, \int g^2 =1, \int |\nabla g|^2\leq C s^{-2}.
	\end{align}
	\begin{enumerate}
		\item For all $V:\mathbb{R}^3\to \mathbb{R}$ such that $[V]_+, [V-V\ast g^2]_+ \in L^{\frac 52}$ and for all density matrices $0\leq \gamma\leq 1$, we have \begin{align}\label{eq:semicl}
		\mathrm{tr}((-\Delta -V)\gamma)\geq -L_{\mathrm{sc}} \int [V]_+^{\frac 52} - Cs^{-2} \mathrm{tr}\gamma-C\left(\int [V]_+^{\frac 52}\right)^{\frac 35}\left(\int [V-V\ast g^2]^{\frac 52}_+ \right)^{\frac 25}.
		\end{align}
		\item On the other hand, if $[V]_+\in L^{\frac 52}\cap L^{\frac 32}$, then there is a density matrix $\gamma$ such that \begin{align}\rho_\gamma=\frac 52 L_{\mathrm{sc}} [V]_+^{\frac 32}\ast g^2\end{align} and \begin{align}
	\mathrm{tr}(-\Delta \gamma) \leq L_{\mathrm{sc}} \int [V]_+^{\frac 52} + Cs^{-2} \int [V]_+^{\frac 32}.
	\end{align}
	\end{enumerate}
\end{lemma}
\subsection{Coulomb potential estimate} The following bound is essentially contained in 
\cite[Corollary~9.3]{solovej2003ionization} and appears explicitly in \cite[Lemma~18]{frank2016ionization}.
	\begin{lemma}\label{eq:coulombest}
		For every $f\in L^{\frac 53}\cap L^{\frac 65}$ and $x\in \mathbb{R}^3$, we have \begin{align}
		\left| \int_{|y|<|x|} \frac{f(y)}{|x-y|} \dd{y} \right|\leq C \|f \|_{\frac 53}^{\frac 56} (|x| D[f])^{\frac{1}{12}}.
		\end{align}
	\end{lemma}

\section{Screened potential estimate}
\begin{lemma}[Screened potential estimate] \label{lem:screenedpot} There are universal constants $C>0, \epsilon>0, D>0$ such that 
\begin{align}
|\Phi_{|x|}(x) - \Phi_{|x|}^{\mathrm{TF}}(x)| \leq C |x|^{-4+\epsilon},\; \forall |x|\leq D.
\end{align}
\end{lemma}
As in \cite{frank2016ionization}, this is proved using a bootstrap argument. 
\begin{lemma}[Initial step] \label{lem:initialstep} There is a universal constant $C_1$ such that 
\begin{align}
|\Phi_{|x|}(x) - \Phi_{|x|}^{\mathrm{TF}}(x)| \leq C_1 Z^{\frac{49}{36} - a} |x|^{\frac{1}{12}}, \; \forall |x|>0
\end{align}
with $a=1/198$.
\end{lemma}
\begin{lemma}[Iterative step]\label{lem:iterativestep} There are universal constants $C_2, \beta, \delta, \epsilon>0$ such that, if 
\begin{align}\label{eq:itcond}
|\Phi_{|x|} (x) -\Phi_{|x|}^{\mathrm{TF}}(x) | \leq \beta |x|^{-4}, \; \forall |x| \leq D
\end{align}
for some $D\in [Z^{-\frac 13}, 1]$, then
\begin{align}
|\Phi_{|x|}(x) - \Phi_{|x|}^{\mathrm{TF}}(x) | \leq C_2 |x|^{-4 + \epsilon}, \; \forall D\leq |x| \leq D^{1-\delta}.
\end{align}
\end{lemma}
Now, we want to prove Lemma~\ref{lem:screenedpot} using Lemma~\ref{lem:initialstep} and Lemma~\ref{lem:iterativestep}.
\begin{proof}[Proof of Lemma~\ref{lem:screenedpot}]
We set $\sigma:=\max\{C_1,C_2\}$. Without loss of generality we may assume that $\beta < \sigma$ and $\epsilon\leq 3a = \frac{1}{66}$. We set 
\begin{align*}
D_n = Z^{-\frac 13 (1-\delta)^n}, \; n=0,1,2,\dots .
\end{align*}
From Lemma~\ref{lem:initialstep}, we have 
\begin{align*}
|\Phi_{|x|}(x) - \Phi_{|x|}^{\mathrm{TF}}(x) | \leq C_1 Z^{\frac{49}{36} - a} |x|^{\frac{1}{12}} \leq \sigma |x|^{-4 +\epsilon},\;  \forall |x|\leq D_0=Z^{-\frac 13}
\end{align*}
and some $\epsilon>0$ small enough. From Lemma~\ref{lem:iterativestep}, we deduced by induction that for all $n=0,1,2,\dots$ if \begin{align*}
\sigma (D_n)^\epsilon \leq \beta,
\end{align*}
then 
\begin{align*}
|\Phi_{|x|}(x) - \Phi^{\mathrm{TF}}_{|x|}(x) | \leq \sigma |x|^{-4+\epsilon}, \; \forall |x| \leq D_{n+1}.
\end{align*}
Note that $D_n\to 1$ as $n\to \infty$ and that $\sigma>\beta$. Thus, there is a minimal $n_0\in \{0,1,2,\dots\}$ such that $\sigma(D_{n_0})^\epsilon > \beta$. If $n_0 \geq 1$, then $\sigma (D_{n_0 - 1})^{\epsilon} \leq \beta $ and therefore, by the preceding argument,
\begin{align*}
|\Phi_{|x|}(x) - \Phi^{\mathrm{TF}}_{|x|} (x) | \leq \sigma |x|^{-4+\epsilon}, \; \forall |x| \leq D_{n_0}.
\end{align*}
As shown above, the same bound holds true for $n_0=0$. Now, let $D=(\sigma^{-1}\beta)^{\frac 1 \epsilon}$, which is an universal constant, and note that by choice of $n_0$ we have $D_{n_0} \geq D$. 
\end{proof}
\section{Initial step}
In this section we prove Lemma~\ref{lem:initialstep}. We write 
\begin{align*}
\mathcal{E}^{\mathrm{RHF}}(\gamma) = \mathcal{E}^{\mathrm{RHF}}_{r=0} (\gamma) = \mathrm{tr} (-\Delta \gamma) - \int \frac{Z \rho_\gamma(x) }{|x|} \dl + D[\rho_\gamma].
\end{align*}
\begin{proof}[Proof of Lemma~\ref{lem:initialstep}]
The strategy is to bound $\mathcal{E}^p_Z(\gamma_0)$ from above and below using semi-classical estimates from Lemma~\ref{lem:semicl}. The main term in both bounds is $\mathcal{E}^{\mathrm{TF}}(\rho^{\mathrm{TF}})$, but in the lower bound we will get an additional term $D[\rho_0-\rho^{\mathrm{TF}}]$. The error terms in the upper and lower bounds will then give an upper bound on $D[\rho_0-\rho^{\mathrm{TF}}]$ which will imply the lemma. 

\textbf{Upper bound.} We shall show that 
\begin{align}\label{eq:upperbound}
\mathcal{E}^p_Z(\gamma_0) \leq \mathcal{E}^{\mathrm{TF}}(\rho^{\mathrm{TF}}) + CZ^{\frac{11}{5}}.
\end{align}
Indeed, by Lemma~\ref{lem:non-inc}, 
\begin{align*}
N\mapsto E^p(N,Z)
\end{align*}
is non-increasing and since the contribution of the exchange term to the energy is non-positive, we have 
\begin{align*}
\mathcal{E}^p_Z(\gamma_0) \leq E^p_{\leq }(N,Z) \leq E^{\mathrm{RHF}}_{\leq} (N,Z),
\end{align*}
where $E^p_{\leq}(N,Z) = \inf\{\mathcal{E}^p_Z(\gamma): \gamma\in \mathcal{I}, \spr{\gamma}\leq N \}$ and analogously for $ E^{\mathrm{RHF}}_{\leq} (N,Z)$.
Now, \eqref{eq:upperbound} follows from a well-known bound on the ground state energy in reduced Hartree-Fock theory \cite[Proof of Theorem~5.1]{lieb1981thomas}.

\textbf{Lower bound.} We now show that 
\begin{align}\label{eq:lowerbound}
\mathcal{E}^p(\gamma_0) \geq \mathcal{E}^{\mathrm{TF}}(\rho^{\mathrm{TF}}) + D[\rho_0 - \rho^{\mathrm{TF}}] - C Z^{\frac{25}{11}}.
\end{align}
With the Thomas-Fermi potential $\varphi^{\mathrm{TF}}= \frac{Z}{|x|} - \rho^{\mathrm{TF}}\ast |x|^{-1}$ we can write 
\begin{align*}
\mathcal{E}^p(\gamma_0) = \mathrm{tr}( (-\Delta - \varphi^{\mathrm{TF}}) \gamma_0) + D[\rho_0 - \rho^{\mathrm{TF}}] - D[\rho^{\mathrm{TF}}] - X(\gamma_0^p).
\end{align*}
Recall from \eqref{eq:aprioribound3} the bound \begin{align} X(\gamma_0^p) \leq CZ^{\frac 53}.
\end{align}
Next, from the semiclassical estimate \eqref{eq:semicl} we have
\begin{align*}
\mathrm{tr}((-\Delta - \varphi^{\mathrm{TF}})\gamma_0)\geq - L_{\mathrm{sc}}\int [\varphi^{\mathrm{TF}}]_+^{\frac 52} - Cs^{-2} \mathrm{tr} \gamma_0 - C \left( \int [\varphi^{\mathrm{TF}}]_+^{\frac 52}\right)^{\frac 35} \left( [\varphi^{\mathrm{TF}} - \varphi^{\mathrm{TF}} \ast g^2]_+^{\frac 52}\right)^{\frac 25}.
\end{align*}
According to \eqref{eq:aprioribound}, we can bound $\mathrm{tr}\gamma_0= N \leq C Z$. Moreover, by scaling, \begin{align*}
\int |\varphi^{\mathrm{TF}}|^{\frac 52} = C \int (\rho^{\mathrm{TF}})^{\frac 53} \leq C Z^{\frac 73}
\end{align*}
and, as explained in \cite[end of page 554]{solovej2003ionization},
\begin{align*}
\int |\varphi^{\mathrm{TF}} - \varphi^{\mathrm{TF}}\ast g^2|^{\frac 52} \leq C Z^{\frac 52} s^{\frac 12}.
\end{align*}
Thus, \begin{align*}\mathrm{tr}((-\Delta -\varphi^{\mathrm{TF}} )\gamma_0) \geq - L_{\mathrm{sc}} \int [\varphi^{\mathrm{TF}}]_+^{\frac 52} - C Z^{\frac{25}{11}}.\end{align*}
Optimising over $s>0$ we get 
\begin{align*}
\mathrm{tr}((-\Delta -\varphi^{\mathrm{TF}}(\gamma_0))\geq -L_{\mathrm{sc}} \int [\varphi^{\mathrm{TF}}]_+^{\frac 52}-CZ^{\frac{25}{11}}.
\end{align*}
From the Thomas-Fermi equation we have \begin{align}\label{eq:tfq}
-L_{\mathrm{sc}} \int [\varphi^{\mathrm{TF}}]_+^{\frac 52} - D[\rho^{\mathrm{TF}}] = \mathcal{E}^{\mathrm{TF}}(\rho^{\mathrm{TF}}),
\end{align}
which proves \eqref{eq:lowerbound}.

\textbf{Conclusion.} Combining \eqref{eq:upperbound} and \eqref{eq:lowerbound}, we deduce that 
\begin{align*}
D[\rho_0-\rho^{\mathrm{TF}}]\leq C Z^{\frac{25}{11}}.
\end{align*}
From the Coulomb estimate \eqref{eq:coulombest} with $f=\rho_0-\rho_0^{\mathrm{TF}}$ and the kinetic estimates
\begin{align*}
\int \rho_0^{\frac 53} \leq C Z^{\frac 73}, \; \int (\rho^{\mathrm{TF}})^{\frac 53} \leq C Z^{\frac 73}
\end{align*}
(the first estimate follows from \eqref{eq:corrho} and the second one from scaling),
we find that for all $|x|>0$, 
\begin{align*}
|\Phi_{|x|} (x) - \Phi_{|x|}^{\mathrm{TF}}(x) | &= \left| \int_{|y|\leq|x|} \frac{\rho_0(y) - \rho^{\mathrm{TF}}(y)}{|x-y|} \dd{y}\right| \\ & \leq C \|\rho_0-\rho^{\mathrm{TF}}\|^{\frac 56}_{\frac 53} ( |x| D[\rho_0 - \rho^{\mathrm{TF}}])^{\frac{1}{12}} \\ &\leq C Z^{\frac{179}{132}} |x|^{\frac{1}{12}}.
\end{align*}
Since $179/132 = 49/36-1/198$, this is the desired bound.
\end{proof}
\section{Iterative step}
The goal of this section is to prove Lemma~\ref{lem:iterativestep}. The proof is split into five steps. 

\textbf{Step 1.} We collect some consequences of \eqref{eq:itcond}.
\begin{lemma}\label{lem:step1}
	Assume that \eqref{eq:itcond} holds true for some $\beta, D\in (0,1]$. Then, for all $r\in (0,D]$, we have 
	\begin{align}
\label{eq:lemmait1}	\left| \int_{|x|<r} (\rho_0 - \rho^{\mathrm{TF}}) \right| & \leq \beta r^{-3},\\
\label{eq:lemmait2}	\sup_{|x|\geq r}	|x||\Phi_r(x)| & \leq C r^{-3}, \\
\label{eq:lemmait3}\int_{|x|>r} \rho_0 & \leq C r^{-3}, \\
\label{eq:lemmait4}\int_{|x|>r }\rho_0^{\frac 53} & \leq C r^{-7}, \\
\label{eq:lemmait5}\mathrm{tr}(-\Delta \eta_r \gamma_0 \eta_r) & \leq C (r^{-7} +  \lambda^{-2} r^{-5}), \; \forall \lambda\in (0,1/2].
	\end{align}
	\begin{proof}
		Let $r\in (0,D]$. By Newton's theorem, we have \begin{align*}
		\int_{|y|<r} (\rho^{\mathrm{TF}}(y) - \rho_0(y) )\dd{y} &= r \int_{\mathbb{S}^2} \left(\int_{|y|<r} \frac{\rho^{\mathrm{TF}}(y) - \rho_0(y)}{|r\nu - y|} \dd{y}\right) \frac{\dd\nu}{4\pi} \\ &= r \int_{\mathbb{S}^2} (\Phi_r(r\nu) - \Phi_r^{\mathrm{TF}}(r\nu) ) \frac{\dd\nu}{4\pi}.
		\end{align*}
		Thus, \eqref{eq:lemmait1} follows directly from \eqref{eq:itcond}.
		
		In order to prove \eqref{eq:lemmait2} we first use the following bounds from Thomas-Fermi theory:
		\begin{align*}
		\varphi^{\mathrm{TF}} (x) \leq C |x|^{-4}, \; \rho^{\mathrm{TF}} (x) \leq C |x|^{-6}.
		\end{align*}
		The first bound is proved in Theorem~5.2 of \cite{solovej2003ionization}. Note that $\mu^{\mathrm{TF}}=0$ since $\rho^{\mathrm{TF}}$ is the minimizer of a neutral atom. The second estimate can be found in the proof of Lemma~5.3 in \cite{solovej2003ionization}. Using these bounds we have for all $|x|>0$
		\begin{align*}
		\Phi_{|x|}^{\mathrm{TF}}(x) = \varphi^{\mathrm{TF}}(x) + \int_{|y|>|x|} \frac{\rho^{\mathrm{TF}}(y)}{|x-y|}\dd{y} \leq C |x|^{-4} ,
		\end{align*}
		where Newton's theorem was used to get the bound on the integral. This implies $\Phi_r^{\mathrm{TF}}(x) \leq C r^{-4}$ for all $|x|=r$. Now, we use the assumptions \eqref{eq:itcond} to obtain
		\begin{align*}
		\Phi_r(x) = ( \Phi_r(x) - \Phi_r^{\mathrm{TF}}(x) )+ \Phi_r^{\mathrm{TF}}(x) \leq C r^{-4} \;\; \forall |x|=r.
		\end{align*}
		Note that $\Phi_r(x)$ is harmonic ($\Delta \Phi_r (x)=0$) for $|x|>r$ and vanishes at infinity. Thus, we can apply Lemma~19 of \cite{frank2016ionization}, which is a consequence of the maximum principle, to obtain
		\begin{align*}
		\sup_{|x|\geq r} |x|\Phi_r(x) = \sup_{|x|=r} |x|\Phi_r(x) \leq C r^{-3}.
		\end{align*}
		Carrying out the same arguments for $-\Phi_r(x)$  gives $ -\sup_{|x|\geq r} |x|\Phi_r(x) \leq C r^{-3}$ which concludes the proof of \eqref{eq:lemmait2}.
		
		Now, we prove \eqref{eq:lemmait3}. Using the assumption \eqref{eq:itcond} and the bound $\rho^{\mathrm{TF}}(x)\leq C |x|^{-6}$, we have 
		\begin{align}
	\nonumber
		\int_{r/3 < |x| < r} \rho_0 &= \int_{|x|< r} (\rho_0 - \rho^{\mathrm{TF}}) - \int_{|x|<r/3} (\rho_0 - \rho^{\mathrm{TF}}) + \int_{r/3 < |x| < r} \rho^{\mathrm{TF}} 
		\\ &\leq \beta r^{-3} + \beta (r/3)^{-3} + C r^{-3} \leq C r^{-3}.\label{eq:proofit}
		\end{align}
		Now, we use Lemma~\ref{lem:extkin} as well as \eqref{eq:lemmait2} and \eqref{eq:proofit} to get
		\begin{align}\nonumber
		\mathrm{tr}(-\Delta \eta_r \gamma_0 \eta_r)\leq &C (1+(\lambda r)^{-2}) \int \chi^+_{(1-\lambda)r}\rho_0 \\ & + C \lambda r^3 \sup_{|z|\geq (1-\lambda) r}[\Phi_{(1-\lambda) r}(z)]_+^{\frac 52} + C \sup_{|z|\geq
			\nonumber r}[|z|\Phi_r(z)]_+^{\frac 73}\\ \leq & C \left( (\lambda r)^{-2} \int \chi_r^+ \rho_0 + \lambda^{-2} r^{-5} + r^{-7} \right).\label{eq:proofitr}
		\end{align}
		Doing the same estimate again but replacing $r$ by $r/3$, we get
		\begin{align}
		\mathrm{tr}(-\Delta \eta_{r/3}\gamma_0 \eta_{r/3}) \leq C \left( (\lambda r)^{-2} \int \chi_r^+ \rho_0 + \lambda^{-2}r^{-5} + r^{-7}\right).\label{eq:proofitr3}
		\end{align}
		From Lemma~\ref{lem:keylemmaext}, replacing $r$ by $r/3$ and choosing $s=r$, we find that 
		\begin{align*}
		\int \chi_{r/3}^+ \rho_0 \leq &C \int_{r/3<  |x|<(1+\lambda)^2 r/3}\rho_0 + C \left(\sup_{|z|\geq r/3}[|z|\Phi_{r/3}(z)]_+ + \lambda^{-2} r^{-1}\right)\nonumber \\ &
		+C \left( r^2 \mathrm{tr}(-\Delta \eta_{r/3}\gamma_0 \eta_{r/3})\right)^{\frac 35} + C \left(r^2 \mathrm{tr}(-\Delta \eta_{r/3}\gamma_0 \eta_{r/3})\right)^{\frac 13}.
		\end{align*}
		Inserting \eqref{eq:lemmait2}, \eqref{eq:proofit} and \eqref{eq:proofitr3} into the latter estimate leads to
		\begin{align*}\nonumber
		\int \chi_r^+\rho_0 \leq \int \chi_{r/3}^+ \rho_0 \leq C (r^{-3} + \lambda^{-2} r^{-1}) &+ C \left( \lambda^{-2} \int \chi_r^+ \rho_0 + \lambda^{-2} r^{-3} + r^{-5}\right)^{\frac 35} \\ & +C \left( \lambda^{-2} \int \chi_r^+ \rho_0 + \lambda^{-2} r^{-3} + r^{-5}\right)^{\frac 13} .
		\end{align*}
		This implies \eqref{eq:lemmait3} (e.g.\ choose $\lambda=1/2$). To obtain \eqref{eq:lemmait5} we just insert \eqref{eq:lemmait3} into \eqref{eq:proofitr}.
		
		We use the kinetic Lieb-Thirring inequality and \eqref{eq:lemmait5} to obtain
		\begin{align*}
		\int_{|x|>r} \rho_0^{\frac 53} \leq \int (\eta_{r/3}\rho_0 \eta_{r/3})^{\frac 53} \leq C \mathrm{tr}(-\Delta \eta_{r/3} \gamma_0 \eta_{r/3}) \leq C (r^{-7} + \lambda^{-2}r^{-5}).
		\end{align*}
		Again, we can choose $\lambda= 1/2$ to get \eqref{eq:lemmait4}.
	\end{proof}
\end{lemma}

	\textbf{Step 2.} Now we introduce the exterior Thomas-Fermi energy functional \begin{align}
	\mathcal{E}_r^{\mathrm{TF}} (\rho)= c^{\mathrm{TF}} \int \rho^{\frac 53} - \int V_r \rho + D[\rho], \;\; V_r = \chi_r^+ \Phi_r.
	\end{align}
	\begin{lemma}\label{lem:step2}
	The functional $\mathcal{E}_r^{\mathrm{HF}}(\rho)$ has a unique minimizer $\rho_r^{\mathrm{TF}}$ over \begin{align*}
	0\leq \rho \in L^{\frac 53}(\mathbb{R}^3)\cap L^{1}(\mathbb{R}^3),\;\; \int \rho \leq \int \chi_r^+ \rho_0.
	\end{align*}
	The minimizer is supported in $\{|x|\geq r\}$ and satisfies the Thomas-Fermi equation
\begin{align*}
\frac{5c^{\mathrm{TF}}}{3}\rho_r^{\mathrm{TF}}(x)^{\frac 23}=[\varphi_r^{\mathrm{TF}}(x) - \mu_r^{\mathrm{TF}}]_+
\end{align*}
with $\varphi_r^{\mathrm{TF}}(x) = V_r - \rho_r^{\mathrm{TF}}\ast |x|^{-1}$ and a constant $\mu_r^{\mathrm{TF}}\geq 0$. Moreover, if \eqref{eq:itcond} holds true for some $\beta, D\in (0,1]$, then 
\begin{align}\label{eq:step2}
\int (\rho_r^{\mathrm{TF}})^{\frac 53} \leq C r^{-7},\;\; \forall r\in (0,D].
\end{align}
	\end{lemma}
	The proof is identical to that of \cite[Lemma 21]{frank2016ionization}.
	
	\textbf{Step 3.} Now we compare $\rho_r^{\mathrm{TF}}$ with $\chi_r^+\rho^{\mathrm{TF}}$.
	\begin{lemma}\label{lem:step3}
		We can choose a universal constant $\beta >0$ small enough such that, if \eqref{eq:itcond} holds true for some $D\in[Z^{-1/3}, 1]$, then $\mu_r^{\mathrm{TF}} =0$ and \begin{align*}
		|\varphi_r^{\mathrm{TF}}(x) - \varphi^{\mathrm{TF}}(x)|& \leq C (r/|x|)^\zeta |x|^{-4}\\
		|\rho_r^{\mathrm{TF}}(x) - \rho^{\mathrm{TF}}(x) |&  \leq C (r/|x|)^\zeta |x|^{-6}
		\end{align*}
		for all $r\in [Z^{-\frac 13},D]$ and for all $|x|>r$. Here $\zeta = (\sqrt{72}-7)/2 \approx 0.77$.
	\end{lemma}
	This proof is also identical to that of \cite[Lemma 22]{frank2016ionization}.

\textbf{Step 4.} In this step, we compare $\rho_r^{\mathrm{TF}}$ with $\chi_r^+\rho_0$. 
\begin{lemma}\label{lem:step4}
	Let $\beta>0$ be as in Lemma~\ref{lem:step3}. Assume that \eqref{eq:itcond} holds true for some $D\in [Z^{-\frac 13},1]$. Then,
	\begin{align*}
	D[\chi_r^+ \rho_0 - \rho_r^{\mathrm{TF}}]\leq C r^{-7+b}, \;\;\;\forall r\in[Z^{-\frac 13},D],
	\end{align*}
	where $b=1/3$.
	\begin{proof}
		The strategy is to bound $\mathcal{E}_r^{\mathrm{RHF}}(\eta_r \gamma_0 \eta_r)$ from above and from below using the semi-classical estimates from Lemma~\ref{lem:semicl}. The main term $\mathcal{E}_r^{\mathrm{TF}}(\rho_r^{\mathrm{TF}})$ will cancel, whereas the additional term $D[\chi_r^+\rho_0 - \rho_r^{\mathrm{TF}}]$ will be bounded by the error terms, which will give the result.
		
		\textbf{Upper bound.} We shall prove that \begin{align}\label{eq:step4upperbound}
		\mathcal{E}_r^{\mathrm{RHF}}(\eta_r\gamma_0\eta_r) \leq \mathcal{E}_r^{\mathrm{TF}}(\rho_r^{\mathrm{TF}}) + C r^{-7}(r^{2/3}+\lambda^{-2}r^2 + \lambda).
		\end{align}
		We use Lemma~\ref{lem:semicl} (ii) with $V_r^\prime =\chi_{r+s}^+ \phi_r^{\mathrm{TF}}$, $s\leq r$ to be chosen later and $g$ spherically symmetric to obtain a density matrix $\gamma_r$ as in the statement. Since $\mu_r^{\mathrm{TF}}=0$ by Lemma~\ref{lem:step3}, we deduce from the Thomas-Fermi equation in Lemma~\ref{lem:step2} that \begin{align*}
		\rho_{\gamma_r} = \frac 52 L_{\mathrm{sc}} \left( \chi_{r+s}^+ (\varphi_r^{\mathrm{TF}})^{\frac 32}\right)\ast g^2 = (\chi_{r+s}^+ \rho_r^{\mathrm{TF}})\ast g^2.
		\end{align*}
		Note that $\rho_{\gamma_r}$ is supported in $\{|x|\geq r\}$ and 
		\begin{align*}
		\mathrm{tr} \gamma_r = \int \rho_{\gamma_r} = \int \chi_{r+s}^+ \rho_r^{\mathrm{TF}} \leq \int \rho_r^{\mathrm{TF}} \leq \int \chi_r \rho_0.
		\end{align*}
		Thus, we may apply Lemma~\ref{lem:insideout} and obtain
		\begin{align}\label{eq:upperboundstep}\mathcal{E}_r^{\mathrm{RHF}}(\eta_r \gamma_0 \eta_r) \leq \mathcal{E}_r^{\mathrm{RHF}}(\gamma_r) + \mathcal{R}.
		\end{align}
		By the semiclassical estimate from Lemma~\ref{lem:semicl} (ii) 
		\begin{align*}
		\mathcal{E}_r^{\mathrm{RHF}} (\gamma_r) \leq& \frac 32 L_{\mathrm{sc}} \int [V_r^\prime]_+^{\frac 52} +  C s^{-2} \int [V_r^\prime]_+^{\frac 32} - \int \Phi_r\left( \chi_{r+s}^+ \rho_r^{\mathrm{TF}} )\ast g^2\right) + D[\rho_r^{\mathrm{TF}}\ast g^2]\\ 
		\leq& \frac 32 L_{\mathrm{sc}} \int [V_r^\prime]_+^{\frac 52} - \int \Phi_r \rho_r^{\mathrm{TF}} + D[\rho_r^{\mathrm{TF}}] + Cs^{-2} \int \rho_r^{\mathrm{TF}} \\ &+\int (\Phi_r - \Phi_r\ast g^2) \chi_{r+s}^+\rho_r^{\mathrm{TF}} + \int_{r\leq |x| \leq r+s} \Phi_r \rho_r^{\mathrm{TF}} 
	\\	=& \mathcal{E}_r^{\mathrm{TF}}(\rho_r^{\mathrm{TF}}) + C s^{-2} \int \rho_r^{\mathrm{TF}} + \int_{r\leq |x|\leq r+s} \Phi_r \rho_r^{\mathrm{TF}},
		\end{align*}
		where we have used the convexity of $D$ in the second inequality. The equality in the last line holds true, since $\Phi_r(x)$ is harmonic when $|x|>r$ and $g$ is chosen spherically symmetric.
		
		According to \eqref{eq:lemmait3} we have
		\begin{align*}
		\int \rho_r^{\mathrm{TF}} \leq \int \chi_r^+ \rho_0 \leq C r^{-3}.
		\end{align*}
		We now use the fact that $\rho_r^{\mathrm{TF}} \leq C |x|^{-6}$ for all $|x|\geq r$, which follows from Lemma~\ref{lem:step3}. Thus,
		\begin{align*}
		\int_{r\leq |x| \leq r+s}\Phi_r\rho_r^{\mathrm{TF}} \leq C r^{-3} \int_{r\leq |x|\leq r+s} |x|^{-1}\rho_r^{\mathrm{TF}} \leq C r^{-8}s,
		\end{align*}
		where we have used \eqref{eq:lemmait2}. Optimising over $s$ (which leads to $s\sim r^{5/3}$) we obtain
		\begin{align}\label{eq:uppberboundstep1}
		\mathcal{E}_r^{\mathrm{RHF}}(\gamma_r)\leq \mathcal{E}_r^{\mathrm{TF}}(\rho_r^{\mathrm{TF}}) + C r^{-7+2/3}.
		\end{align}
		To estimate $\mathcal{R}$ we use Lemma~\ref{lem:step1} and obtain 
		\begin{align*}
		\mathcal{R} \leq C(1+ (\lambda r)^{-2})r^{-3} + C \lambda^3 (r^{-4})^{5/2} + C (r^{-7} + 
		\lambda^{-2} r^{-5})^{\frac 12} (r^{-3})^{\frac 12} 
\\	\leq C (\lambda^{-2}r^{-5} + \lambda r^{-7}).
		\end{align*}
		Combining this with \eqref{eq:upperboundstep} and \eqref{eq:uppberboundstep1} we get the upper bound \eqref{eq:step4upperbound}.
		
		\textbf{Lower bound.} We shall prove that 
		\begin{align}\mathcal{E}_r^{\mathrm{RHF}}(\eta_r \gamma_0\eta_r )\geq \mathcal{E}_r^{\mathrm{TF}}(\rho_r^{\mathrm{TF}}) +  D[\eta_r^2 \rho_0 - \rho_r^{\mathrm{TF}}] - C r^{-7 + 1/3}.\label{eq:step4lowerbound}
		\end{align}
		We use Lemma~\ref{lem:semicl} (i) in a way similar to the proof of Lemma~\ref{lem:initialstep} to obtain
		\begin{align*}\mathcal{E}_r^{\mathrm{RHF}}(\eta_r \gamma_0 \eta_r) = \,& \mathrm{tr} ((-\Delta -\varphi_r^{\mathrm{TF}})\eta_r\gamma_0\eta_r) + D[\eta_r^2 \rho_0 - \rho_r^{\mathrm{TF}}]-D[\rho_r^{\mathrm{TF}} ] \\ 
		\geq & - L_{\mathrm{sc}} \int [\varphi_r^{\mathrm{TF}}]_+^{\frac 52} - Cs^{-2} \int \eta_r^2 \rho_0 \\
		- & C \left( \int [\varphi_r^{\mathrm{TF}}]^{\frac 52}_+\right)^{\frac 35} \left(\int [\varphi_r^{\mathrm{TF}} - \varphi_r^{\mathrm{TF}}\ast g^2]_+^{\frac 52} \right)^{\frac 25} \\
		& + D[\eta_r^2 \rho - \rho_r^{\mathrm{TF}}] - D[\rho_r^{\mathrm{TF}}] \\
		= &\mathcal{E}_r^{\mathrm{TF}}(\rho_r^{\mathrm{TF}}) + D[\eta_r^2 \rho_0 - \rho_r^{\mathrm{TF}}]  - Cs^{-2} \int \eta_r^2 \rho_0 
		\\ & - C \left( \int [\varphi_r^{\mathrm{TF}}]_+^{\frac 52}\right)^{\frac 35} \left( \int [\varphi_r^{\mathrm{TF}} - \varphi_r^{\mathrm{TF}}\ast g^2 ]_+^{\frac 52}\right)^{\frac 25}.
		\end{align*} 
		The last identity was derived using the Thomas-Fermi equations similarly as in \eqref{eq:tfq}. In order to control the remainder terms, by Lemma~\ref{lem:step1} and Lemma~\ref{lem:step2}, we have
		\begin{align*}
		\int \eta_r^2 \rho_0 \leq C r^{-3}, \; \int [\varphi_r^{\mathrm{TF}}]^{\frac 52} = C \int (\rho_r^{\mathrm{TF}}) 
		\leq C r^{-7}.
		\end{align*}
		In order to bound the convolution term we use - as in the proof of Lemma~\ref{lem:initialstep} - the fact that $|x|^{-1} - |x|^{-1}\ast g^2 \geq 0$, and therefore also $\rho_r^{\mathrm{TF}}\ast (|x|^{-1}-|x|^{-1}\ast g^2)\geq 0$. Since $\varphi_r^{\mathrm{TF}} = \chi_r^+ \Phi_r - \rho_r^{\mathrm{TF}} \ast |x|^{-1}$, we conclude that 
		\begin{align*}
		\phi_r^{\mathrm{TF}} - \phi_r^{\mathrm{TF}} \ast g^2  \leq \chi_r^+ \Phi_r - (\chi_r^+ \Phi_r) \ast g^2.
		\end{align*}
		Since $\Phi_r$ is harmonic outside a ball of radius $r$ and $g$ is spherically symmetric, $\chi_r^+ \Phi_r - (\chi_r^+ \Phi_r) \ast g^2$ is supported in $\{ r-s \leq |x| \leq r+s\}$ and, by Lemma~\ref{lem:step1}, its absolute value is bounded by $Cr^{-4}$. Thus, 
		\begin{align*}
		[\varphi_r^{\mathrm{TF}}-\varphi_r^{\mathrm{TF}}\ast g^2]_+ \leq C r^{-4} \mathbb{1}(r-s\leq |\cdot| \leq r+s)
		\end{align*}
		and therefore,
		\begin{align*}
		\int [\varphi_r^{\mathrm{TF}} - \varphi_r^{\mathrm{TF}}\ast g^2]_+^{\frac 52} \leq C r^{-8} s.
		\end{align*}
		To summarize, we have shown that 
		\begin{align*}
		\mathcal{E}_r^{\mathrm{RHF}}(\eta_r \gamma_0 \eta_r) \geq \mathcal{E}_r^{\mathrm{TF}}(\rho_r^{\mathrm{TF}}) + D[\eta_r^2\rho_0 - \rho_r^{\mathrm{TF}}] - C(s^{-2}r^{-3} + r^{-37/5}s^{2/5}).
		\end{align*}
		Optimising over $s$ (leading to $s\sim r^{11/6}$) we obtain \eqref{eq:step4lowerbound}.
		
		\textbf{Conclusion} Combining \eqref{eq:step4upperbound} and \eqref{eq:step4lowerbound} we infer that \begin{align}
		D[\eta_r^2 \rho_0 - \rho_r^{\mathrm{TF}}] \leq C r^{-7} (r^{\frac 13} + \lambda^{-2} r^2 + \lambda).
		\end{align}
		Now, we want to replace $\eta_r^2$ by  $\chi_r^+$. Using the Hardy-Littlewood-Sobolev inequality and \eqref{eq:lemmait4}, we get \begin{align*}
		D[\chi_r^+\rho_0 - \eta_r^2 \rho_0] & \leq D[\mathbb{1}( (1+\lambda)r \geq |x| \geq r) \rho_0]\\
		&\leq C \|\mathbb{1}( (1+\lambda)r \geq |x| \geq r) \rho_0\|_{L^{\frac 65}}^2\\
		&\leq C \left( \int \chi_r^+ \rho_0^{\frac 53} \right)^{\frac 65} \left( \int_{(1+\lambda)r\geq |x|\geq r}\right)^{\frac{7}{15}}\\
		& \leq C (r^{-7})^{\frac 65} (\lambda r^3)^{\frac{7}{15}} = C \lambda^{\frac{7}{15}}r^{-7}.
		\end{align*}
		Therefore, \begin{align*}
		D[\chi_r^+\rho_0 - \rho_r^{\mathrm{TF}}] \leq 2 D[\chi_r^+ \rho_0 - \eta_r^2 \rho_0] + 2 D[\eta_r^2 \rho_0 - \rho_r^{\mathrm{TF}}] \leq C r^{-7} \left(\lambda^{\frac{7}{15}} + r^{\frac 13} + \lambda^{-2} r^2\right).
		\end{align*}
		This bound is valid for all $\lambda \in (0,1/2]$ and by optimising over $\lambda$ (leading to $\lambda\sim r^{30/37}$) we obtain
		\begin{align*}
		D[\chi_r^+\rho_0 - \rho_r^{\mathrm{TF}}]\leq C r^{-7 + 1/3}.
		\end{align*}
	\end{proof}
\end{lemma}
	\textbf{Step 5.} We are now in the position to prove Lemma~\ref{lem:iterativestep}.
	\begin{proof}[Proof of Lemma~\ref{lem:iterativestep}]
		Let $r\in [Z^{-\frac 13}, D]$ and $|x|\geq r$. As in \cite[Eq. (97)]{solovej2003ionization}, we can decompose 
		\begin{align*}
		\Phi_{|x|}(x) - \Phi_{|x|}^{\mathrm{TF}}(x) = \varphi_r^{\mathrm{TF}}(x) - \varphi^{\mathrm{TF}}(x) + \int_{|y|> |x|} \frac{\rho_r^{\mathrm{TF}}(y) - \rho^{\mathrm{TF}}(y)}{|x-y|}\dd{y} \\ + \int_{|y|<|x|} \frac{\rho_r^{\mathrm{TF}}(y) - (\chi_r^+ \rho_0)(y)}{|x-y|}\dd{y}.
		\end{align*}
		By Lemma~\ref{lem:step3}, we have \begin{align*}
		|\varphi_r^{\mathrm{TF}}(x) - \varphi^{\mathrm{TF}}(x)| \leq C \int_{|y|>|x|} \frac{(r/|y|)^\zeta |y|^{-6}}{|x-y|}\dd{y} \leq C (r/|x|)^\zeta |x|^{-4}.
		\end{align*}
		Moreover, from \eqref{eq:coulombest}, \eqref{eq:lemmait4}, \eqref{eq:step2} and Lemma~\ref{lem:step4}, we get 
		\begin{align*}
		\left| \int_{|y|<|x|} \frac{\rho_r^{\mathrm{TF}}(y) - (\chi_r^+\rho_0)(y)}{|x-y|} \dd{y}\right| &\leq C \| \rho_r^{\mathrm{TF}} - \chi_r^+\rho_0\|_{L^{5/3}}^{5/6} \left( |x| D[\rho_r^{\mathrm{TF}} - \chi_r^+ \rho_0]\right)^{\frac{1}{12}} \\ &\leq C (r^{-7})^{1/2} (|x|r^{-7 + b})^{\frac{1}{12}} \\ &= C |x|^{-4 +b/12} (|x|/r)^{4+1/12-b/12}.
		\end{align*}
		Thus, in summary, for all $r \in [Z^{-1/3}, D]$ and $|x|\geq r$, we have \begin{align}\label{eq:toconcl}
		|\Phi_{|x|}(x) - \Phi_{|x|}^{\mathrm{TF}} (x) | \leq C (r/|x|)^\zeta |x|^{-4} + C (|x|/r)^5 |x|^{-4 + b/12}.
		\end{align}
		With \eqref{eq:toconcl} we will conclude now. First, we choose a $\delta \in (0,1)$ sufficiently small such that \begin{align}\label{eq:1delta}
		\frac{1+\delta}{1-\delta}\left(\frac{49}{36} - a\right) < \frac{49}{36}
		\end{align}
		and \begin{align}\label{eq:b12}
		\frac{b}{12}- \frac{10\delta}{1-\delta} > 0.
		\end{align}
		Here, $a$ and $b$ are the constants from Lemma~\ref{lem:initialstep} and \ref{lem:step4}, respectively. Now, we have two cases.
		
		\textbf{Case 1:} $D^{1+\delta} \leq Z^{-\frac 13 }$. In this case, we simply use the initial step. Indeed, for all \begin{align*}
		|x|\leq D^{1-\delta} \leq (Z^{-\frac 13})^{\frac{1-\delta}{1+\delta}},
		\end{align*}
		by Lemma~\ref{lem:initialstep} we have
		\begin{align}\label{eq:casee1}
		|\Phi_{|x|}(x) - \Phi_{|x|}^{\mathrm{TF}}(x)| \leq C_1 Z^{49/36 -a} |x|^{\frac{1}{12}} \leq C_1 |x|^{1/12 - 3 \frac{1+\delta}{1-\delta} (49/36 -a)}.
		\end{align}
		Note that \begin{align*}
		\frac{1}{12} - 3 \cdot \frac{49}{36} = -4.
	\end{align*}
	Therefore, \eqref{eq:1delta} implies that \begin{align*}
	\frac{1}{12}- \frac{3(1+\delta)}{1-\delta}\left(\frac{49}{36} -a\right) > -4.
	\end{align*}
	\textbf{Case 2:} $D^{1+\delta} \geq Z^{-\frac 13}$. In this case, we use \eqref{eq:toconcl} with $r=D^{1+\delta}$. For all $D\leq |x| \leq D^{1-\delta}$ we have 
	\begin{align*}
	|x|^{2\delta / (1-\delta)} \leq \frac{r}{|x|} \leq |x|^\delta .
	\end{align*}
	Hence, \eqref{eq:toconcl} implies that \begin{align}
	\label{eq:case2}
		|\Phi_{|x|}(x) - \Phi_{|x|}^{\mathrm{TF}}(x) | \leq C |x|^{-4+\zeta \delta } + C |x|^{-4 + b/12 -10\delta/(1-\delta)}.
	\end{align}
	Both exponents of $|x|$ are strictly greater than $-4$ according to \eqref{eq:b12}. 
	
	In summary, from \eqref{eq:casee1} and \eqref{eq:case2}, we conclude that in both cases, 
	\begin{align*}
			|\Phi_{|x|}(x) - \Phi_{|x|}^{\mathrm{TF}}(x) \leq C |x|^{-4+\epsilon}, \;\; \forall D \leq |x| \leq D^{1-\delta}
	\end{align*}
	with \begin{align*}\epsilon:= \min\left\{ \frac{1}{12} - \frac{3(1+\delta)}{1-\delta} (\frac{49}{36} -a) +4, \frac{b}{12}-\frac{10\delta}{1-\delta}, \zeta \delta \right\} >0.
	\end{align*}
	This finishes the proof of Lemma~\ref{lem:iterativestep}.
	\end{proof}
\section{Proof of the main theorems}
\begin{proof}[Proof of Theorem~\ref{thm:ionization}]
	Since we have already proved $N\leq 2Z + C(Z^{\frac 23} + 1)$, we are left with the case $N\geq Z \geq 1$. By Lemma~\ref{lem:screenedpot}, we find universal constants $C, \epsilon, D >0$ such that 
	\begin{align}
	|\Phi_{|x|}(x) - \Phi_{|x|}^{\mathrm{TF}}(x)| \leq C |x|^{-4+\epsilon} , \;\;\forall |x|\leq D.
\end{align}
In particular, \eqref{eq:itcond} holds true with an universal constant $\beta = C D^\epsilon$. We can choose $D$ small enough such that $D\leq 1$ and $\beta \leq 1$, which allows us to apply Lemma~\ref{lem:step1}. Then, using \eqref{eq:lemmait1} and \eqref{eq:lemmait3} with $r=D$, we find that
\begin{align*}
\int_{|x|>D} \rho_0 + \left| \int_{|x|<D} (\rho_0 -\rho^{\mathrm{TF}}) \right| \leq C.
\end{align*}
Since $\int \rho^{\mathrm{TF}} = Z$ we obtain the ionization bound
\begin{align*}N =\int \rho_0 = \int_{|x|>D} \rho_0 + \int_{|x|<D} (\rho_0 - \rho^{\mathrm{TF}}) + \int_{|x|<D}\rho^{\mathrm{TF}} \leq C + Z.
\end{align*}
\end{proof}
\begin{proof}[Proof of Theorem~\ref{thm:mainthm2}]
	 By Lemma~\ref{lem:screenedpot}, we find universal constants $C, \epsilon, D >0$ such that 
	 \begin{align*}
	 |\Phi_{|x|}(x) - \Phi_{|x|}^{\mathrm{TF}}(x)| \leq C |x|^{-4+\epsilon} , \;\;\forall |x|\leq D.
	 \end{align*}
	 As before, we can assume $D\leq 1$ and $CD^\epsilon\leq 1$ in order to apply Lemma~\ref{lem:screenedpot}. 
	 
	 Thus, we are left with the case $|x|>D$. For this we decompose \begin{align}
	 \Phi_{|x|}(x) - \Phi^{\mathrm{TF}}_{|x|}(x)  = \Phi_D (x) - \Phi_D^{\mathrm{TF}} + \int_{|x|> |y|>D} \frac{\rho^{\mathrm{TF}}(y) - \rho_0(y)}{|x-y|}\dd{y}. \label{eq:pfthm2}
	 \end{align}
	 Since $\Phi_{|x|} (x)- \Phi_{|x|}^{\mathrm{TF}} (x) $ is harmonic for $|x|>D$ and vanishes at infinity, we can apply Lemma~19 of \cite{frank2016ionization} to find that 
	 \begin{align*}
	 \sup_{|x|\geq D} |\Phi_D(x) - \Phi_D^{\mathrm{TF}}(x) | = \sup_{|x|= D} |\Phi_D(x) - \Phi_D^{\mathrm{TF}}(x) | \leq C D^{-4+\epsilon}.
	 \end{align*}
	 Moreover, using the bound $\rho^{\mathrm{TF}}(y)\leq C |y|^{-6}$, we can estimate 
	 \begin{align*}
	 \int_{|x|>|y|>D} \frac{\rho_0(y)}{|x-y|}\dd{y} \leq C \int_{|x|>|y|>D} \frac{|y|^{-6}}{|x-y|}\dd{y} \leq C D^{-4}.
	 \end{align*}
	 Finally, using \eqref{eq:lemmait3} and \eqref{eq:lemmait4}, we have \begin{align*}
	 \int_{|x|>|y|>D } \frac{\rho_0(y)}{|x-y|} &\leq  \int_{|y|>D, |x-y|>D } \frac{\rho_0(y)}{|x-y|} +  \int_{|y|>D\geq |x-y| } \frac{\rho_0(y)}{|x-y|} \\& \leq  \int_{|y|>D } 
	 \frac{\rho_0(y)}{D} + \left( \int_{|y|>D} \rho_0(y)^{5/3} \dd{y}\right)^{3/5} \left( \int_{D\geq |x-y|} \frac{1}{|x-y|^{5/2}} \dd{y}\right)^{2/5}\\ & \leq CD^{-4} + C (D^{-7})^{3/5} (\sqrt D)^{2/5}\leq C D^{-4}.
	 \end{align*}
	 Thus, from \eqref{eq:pfthm2} we conclude that \begin{align*}
	 |\Phi_{|x|}(x) - \Phi_{|x|}^{\mathrm{TF}} (x) | \leq C D^{-4}, \;\; \forall |x| >D.
	 \end{align*}
	 In summary, 
	 \begin{align*}
	  |\Phi_{|x|}(x) - \Phi_{|x|}^{\mathrm{TF}} (x) | \leq C |x|^{-4 +\epsilon} + C D^{-4}, \; \forall |x|>0
	 \end{align*}
	 which concludes the proof.
\end{proof}
\begin{proof}[Proof of Theorem~\ref{thm:mainthm3}]
As before, we start by using Lemma~\ref{lem:screenedpot} to find universal constants $C,\epsilon, D>0$ such that \begin{align*}
	 |\Phi_{|x|}(x) - \Phi_{|x|}^{\mathrm{TF}}(x)| \leq C |x|^{-4+\epsilon} , \;\;\forall |x|\leq D.
\end{align*}
We assume that $\epsilon \leq \zeta, D\leq 1$ and $C D^\epsilon \leq 1$. (Recall $\zeta = (\sqrt{73} -7)/2 \approx 0.77$.) From \eqref{eq:lemmait4} and the bound $N\leq Z+C$ in Theorem~\ref{thm:ionization}, we get for all $r\in (0,D]$, 
	\begin{align*}
	\left |  \int_{|y|\geq r|} (\rho_0(y) - \rho^{\mathrm{TF}}(y) ) \dd{y} \right| = \left| N-Z -\int_{|y|< r} (\rho_0(y) - \rho^{\mathrm{TF}}(y)) \dd{y}\right| \leq C r^{-3 +\epsilon}.
	\end{align*}
	From \cite[Theorem~11]{frank2016ionization}, we have 
	\begin{align*}
	\left|\rho^{\mathrm{TF}}(x) - \left( \frac{3A^{\mathrm{TF}}}{5 c^{\mathrm{TF}}} \right)^{3/2}  |x|^{-6} \right| \leq C |x|^{-6}  \left( \frac{Z^{-1/3}}{|x|}\right)^\zeta , \;\; \forall |x| \geq Z^{-1/3}
	\end{align*}
	with $A^{\mathrm{TF}} = (5 c^{\mathrm{TF}})^3 (3\pi^2)^{-1}$.
	Inserting this in the latter estimate over $|x| > r \geq Z^{-1/6}$ and using 
	\begin{align*}
	\left( \frac{Z^{-1/3}}{|x|}\right)^\zeta \leq (r^2 /r )^{\zeta }= r^\zeta \leq r^\epsilon
	\end{align*}
	we obtain 
	\begin{align*}
	\left| \int_{|x|>r} \rho^{\mathrm{TF}}(x) \dl - (B^{\mathrm{TF}}/r)^3 \right| \leq C r^{-3+\epsilon}, \;\; \forall r\in [Z^{-1/6}, D],\end{align*}
	where $B^{\mathrm{TF}} = 5 c^{\mathrm{TF}} (4/(3\pi^2))^{1/3}$. Hence, 
	\begin{align}\label{eq:pfmain3}
\left | \int_{|x|>r}  \rho_0(x) \dl - (B^{\mathrm{TF}}/r)^3\right| \leq C r^{-3+\epsilon}, \;\; \forall r \in [Z^{-1/6 + \epsilon}].
	\end{align}
	Applying \eqref{eq:pfmain3} with $r=D$ and $r= Z^{-1/6}$ yields 
	\begin{align*}
	\int_{|x|>D} \rho_0(x) \dl \leq CD^{-3} , \;\; \int_{|x|> Z^{-1/6}} \rho_0(x)\dl \geq C^{-1} Z^{1/2}.
	\end{align*}
	Thus, if we restrict to the case $C^{-1}Z^{1/2} > \kappa > CD^{-3}$, $R_\kappa := R(N,Z,\kappa) \in [Z^{1/6}, D]$ we can apply \eqref{eq:pfmain3} with $r=R_k$. We obtain 
	\begin{align*}
	\left| \kappa - (B^{\mathrm{TF}}/R_\kappa)^3 \right| \leq C R_\kappa^{-3+\epsilon}.
	\end{align*}
	Setting $t:= \kappa^{1/3} R_\kappa /B^{\mathrm{TF}}$, we can write this as 
	\begin{align*}
	|t^3-1| \leq C (t\kappa^{-1/3})^\epsilon.
	\end{align*}
	Using $|t-1|= \frac{|t^3-1|}{t^2 + t + 1 }\leq \frac{|t^3 -1 |}{t^\epsilon}$ we conclude that 
	\begin{align*}
	|\kappa^{1/3} R_\kappa/B^{\mathrm{TF}} -1 | \leq C \kappa^{-\epsilon/3}.
	\end{align*}
	Thus, if $\kappa > CD^{-3}$, then \begin{align*}\limsup_{N\geq Z \to \infty} |\kappa^{1/3} R_\kappa/B^{\mathrm{TF}} -1 | \leq C \kappa^{-\epsilon /3},
	\end{align*}
which is equivalent to the desired estimate.
\end{proof}
\bibliography{ion_conj_power_v1.bib}
\bibliographystyle{plain}
\end{document}